\documentclass{article}
\usepackage[T1]{fontenc}
\usepackage{geometry}
\usepackage{graphicx} 
\usepackage{amsmath,amssymb}
\usepackage{amsthm}
\usepackage{enumerate}
\usepackage{subcaption}
\usepackage{tikz}
\usepackage{placeins}
\usepackage{comment}
\usepackage[color=white]{todonotes}
\usepackage{cite}
\usepackage[colorlinks=true,citecolor=green!70!black]{hyperref}
\usepackage[nameinlink]{cleveref}
\usepackage{doi}

\usetikzlibrary{decorations.markings}
\usetikzlibrary{arrows}
\usetikzlibrary{arrows.meta}

\newcommand{\lca}{\textsf{LCA}}
\newcommand{\ch}{\textsf{ch}}

\newcommand{\internal}[1]{\overset{\circ}{#1}}

\setuptodonotes{inline}

\newtheorem{theorem}{Theorem}
\newtheorem{lemma}{Lemma}

\newtheoremstyle{thmpart}
{3pt}
{3pt}
{}
{}
{\normalfont\bfseries}
{:}
{.5em}
{\thmname{#1}\thmnumber{ #2}\thmnote{ (#3)}}

\theoremstyle{thmpart}
\newtheorem{definition}{Definition}
\newtheorem{myclaim}{Claim}

\crefname{claim}{Claim}{Claims}
\crefname{enumi}{Type}{Types}

\newcommand{\mytype}[1]{\hyperlink{fig:#1}{\Cref*{#1}}}

\tikzset{>={Stealth[width=4pt]}}

\author{Cyriac Antony, Alessio Martino, and Blerina Sinaimeri\\
Department of AI, Data and Decision Sciences, Luiss University, Rome, Italy\\
\texttt{\{acyriac,amartino,bsinaimeri\}@luiss.it}}

\date{}

\providecommand{\keywords}[1]
{
  \small	
  \textbf{\textit{Keywords:}} #1
}

\begin{document}
%
\title{Structural properties of\\ distance-bounded phylogenetic reconciliation}
%
%

%
\maketitle              
\begin{abstract}
Phylogenetic reconciliation seeks to explain host-symbiont co-evolution by mapping parasite trees onto host trees through events such as cospeciation, duplication, host switching, and loss. Finding an optimal reconciliation that ensures time-consistency is computationally hard when timing information is incomplete, and the complexity remains open when host switches are restricted by a fixed maximum distance $d$. While the case $d=2$ is known to be polynomial, larger values are unresolved. In this paper, we study the cases  $d=3$ and $d=4$. We show that although arbitrarily large cycles may occur, it suffices to check only bounded-size cycles (we provide a complete list), provided the reconciliation  satisfies acyclicity (i.e., time-consistency) in a stronger sense, which we call strong acyclicity. 
Strong acyclicity is associated with incorporating discrete time slices into the model. 

\end{abstract}
\section{Introduction}
Understanding how parasites (or more generally symbionts) have co-evolved with their hosts is a fundamental problem which has applications that range from evolutionary biology to public health, for example in uncovering the origins of new infectious diseases \cite{Etherington2006,Lei2014,PENNINGTON2015}. The increasing availability of DNA sequence data has only amplified the importance of such analyses. One of the main approaches to study host–symbiont co-evolution is through \textit{cophylogeny},  which derive joint evolutionary scenarios for both hosts and parasites, typically represented as phylogenetic trees computed from sequence data.   In this framework, the interaction between host and symbiont lineages is captured by a \textit{reconciliation}: a mapping of the parasite tree onto the host tree (see, e.g., \cite{Menet2022,Charleston2003,MM2005,Donati2015,stolzer12}). This mapping allows the identification of four macro evolutionary events: (a) \emph{cospeciation}, where parasite divergence coincides with host divergence; (b) \emph{duplication}, where the parasite diverges independently of the host; (c) \emph{host switching}, where a parasite colonizes a new host species unrelated to host divergence; and (d) \emph{loss}, which can capture cases where the host speciates but the parasite lineage follows only one of the new host species.

Reconciliations can be constructed in a parsimonious framework: one assigns a cost to each of the four types of events and then seeks to minimize the total cost (computed in an additive way).   If timing information (i.e., the order in which the speciation events occurred in the host phylogeny for incomparable vertices) is not known (or is considered as not sufficiently reliable to be used), as is often the case,  the problem is NP-hard \cite{OFCLH2011,Tofigh11}.  A way to deal with this is either to rely on heuristics \cite{CFOLH2010} (that is, on approaches that are not guaranteed to be optimal) or to settle for solutions that may be biologically infeasible due to time-inconsistency, that is solutions where some of the switches induce a contradictory time ordering for the internal vertices of the host tree \cite{Bansal12,stolzer12,Donati2015,Jacox2016,capybara}.  It is worth noting that while finding an optimal time-consistent reconciliation is computationally hard, checking whether a given reconciliation is time-consistent can be done in polynomial time by constructing an auxiliary directed graph encoding the temporal constraints and verifying the absence of cycles (see e.g. \cite{stolzer12,capybara}).

Because timing information is often unavailable, one alternative is to enrich the reconciliation framework with other biological signals. An important example is the distance of host switches, defined as the distance in the host tree between the donor and recipient species of a switch. This measure reflects the idea that symbionts more readily colonize closely related hosts, while long-distance switches may require substantial adaptive changes \cite{DeVienne2007,Poulin2003}. Distance has been integrated into parsimony-based frameworks in different ways. For instance, \cite{Donati2015} restrict switches to occur only within a given bounded distance in the host tree, whereas in \cite{Bansal2013}, the cost of a switch increases with its distance so that longer switches are penalized but not forbidden. In both cases, however, there is no guarantee that the resulting optimal reconciliations are time-consistent.

It therefore remains an open problem to determine whether fixing the maximum allowed distance $d$ (when $d$ is not part of the input) could yield optimal, time-consistent reconciliations in polynomial time. 
For the special case $d=2$, the problem is solvable in polynomial time since all reconciliations are acyclic as pointed out in \cite{Tavernelli22}; but as the authors emphasize, this result does not extend directly to larger values of $d$. 
In this paper, we address the next open cases, namely $d=3$ and $d=4$. 
Note that arbitrarily large cycles may arise in a reconciliation (e.g., start from a node \( v \)  which is adjacent to a leaf, then repeatedly move to the sibling node of the parent of the current node until it is no longer possible, then move to the sibling node, and then follow arcs in \( H \) to reach back to \( v \)). 
We show that despite this, it is nevertheless sufficient to focus on cycles of bounded size (we provide a complete list of such short cycles), provided the reconciliation  satisfies acyclicity in a stronger sense. 
This makes it possible to check for cycles locally, which can have algorithmic consequences. 
While these results do not directly resolve the algorithmic complexity of the problem, they shed new light on its structural properties and provide useful tools for further studies.

The paper is organized as follows. 
\Cref{sec:prelim} presents some preliminaries. 
\Cref{sec:acyclicity} briefly explains acyclicity, and introduces strong acyclicity. 
\Crefrange{sec:d=3}{sec:d4} are dedicated to the cases \( d=3 \) and \( d=4 \), respectively. 
\Cref{sec:what for algo} discusses potential algorithmic consequences. 
We conclude with \Cref{sec:open} on open problems. 
%
%
%
%
%

\section{Preliminaries}\label{sec:prelim}
In this section we provide some definitions and basic results that will be used throughout the paper.

For a directed graph $G$, we denote by $V(G)$ and $E(G)$ respectively the set of nodes and the set of arcs of $G$. The out-neighbors of a node $v$ are called its children. We consider ordered rooted trees in which arcs are directed away from the root. For a tree $T$, we denote by $L(T)$ the set of leaf nodes, i.e.  those nodes without children, and denote by $r_T$ the root of $T$; the non-leaf nodes are called the internal nodes of $T$. 
A rooted full binary tree is a rooted tree in which every internal node has two children.  A \textit{phylogenetic tree} is a rooted full binary tree. 

We denote by $p(w)$ the parent of a node $w$. 
If there exists a path from a node $v$ to a node $w$, the node $w$ is called a \emph{descendant} of $v$, and $v$ is called an \emph{ancestor} of $w$; if moreover $v\neq w$, we say that $w$ is a \emph{proper descendant} of $v$, and that $v$ is a \emph{proper ancestor} of $w$. If neither $w$ is an ancestor of $v$ nor $v$ is an ancestor of $w$, we say that the two nodes are \emph{incomparable}, and denote this as $v\not\sim w$.  

If a node \( u \) is a child of a child of \( v \), then \( u \) is said to a \emph{grandchild} of \( v \). 
If a node \( u \) and a node \( v \) have the same parent, then \( u \) and \( v \) are \emph{siblings}. 
If the parent of a node \( u \) and the parent of a node \( v \) are siblings, then \( u \) and \( v \) are said to be \emph{cousins}. 

We denote by $\lca(v,w)$ the lowest common ancestor of two nodes $v$ and $w$. The subtree of a tree $T$ rooted at a node $v$ containing all descendants of $v$ is denoted by $T|_v$. 
Throughout this paper, for a node \( v \) in the host tree \( H \), we denote the vertex set of \( H|_v \) by \( S_v \). 
We denote by $d_T(v,w)$ the distance, i.e., the number of arcs on a path between two nodes $v$ and $w$ in $T$.
We define the `levels' of nodes in a rooted tree in such a way that nodes closer to the root are at the higher level; formally, we define \emph{level} of a node \( v \) in a rooted tree \( T \) with root \( r \) as \( -d_T(v,r) \). 
We now introduce the notion of \textit{reconciliation}. 

Let $H$ and $P$ be respectively the rooted phylogenetic trees of the host and parasite species, both binary and full. Let $\sigma$ be a function from $L(P)$ to $L(H),$ representing the parasite/host associations between extant species. A reconciliation is a function $\varphi$ that assigns, for each parasite node $p\in V(P)$, a host node $\varphi(p)\in V(H)$, and satisfies the conditions stated in \Cref{defn:recon} below. A reconciliation must induce an event function $E_\varphi$ on $V(P)$ which associates each parasite node $p$ to an event $E_\varphi(p)$. The set of events is denoted by $\mathcal{E} := \{\mathbb{C} ,\,\mathbb{D},\,\mathbb{S},\,\mathbb{T}\}$; the leaf parasite node has a special event $\mathbb{T}$; for internal parasite nodes, the event $E_\varphi(p)$ is one among three options: \emph{cospeciation} $\mathbb{C}$, \emph{duplication} $\mathbb{D}$, and \emph{host-switch} $\mathbb{S}$. The event for an internal node $p$ will depend on the hosts that are assigned by $\varphi$ to $p$ and to the two children $p_1$ and $p_2$ of $p$. In \Cref{defn:recon}, this dependency is expressed by $E_\varphi(p) := E(\varphi(p),\varphi(p_1),\varphi(p_2))$.

\begin{definition}[Reconciliation, Events of a parasite node]\label{defn:recon}
Given two phylogenetic trees $H$ and $P$, and a function $\sigma: L(P)\to L(H)$, a \emph{reconciliation} of $(H, P, \sigma)$ is a function $\varphi: V(P)\to V(H)$ satisfying the following:
\begin{enumerate}
\item For every leaf node $p\in L(P),$  $\varphi(p)$ is equal to $\sigma(p)$, and $E_\varphi(p)=\mathbb{T}$.
\item For every internal node $p\in V(P)\setminus L(P)$ with children $(p_1,\,p_2),$ exactly one of the following applies:
\begin{enumerate}
\item $E\left(\varphi(p),\varphi(p_1),\varphi(p_2)\right)=\mathbb{S}$, that is, either $\varphi(p_1)\not\sim\varphi(p)$ and $\varphi(p_2)$ is a descendant of $\varphi(p)$, or $\varphi(p_2)\not\sim\varphi(p)$ and $\varphi(p_1)$ is a descendant of $\varphi(p)$,
\item $E\left(\varphi(p),\varphi(p_1),\varphi(p_2)\right)=\mathbb{C}$, that is, $\lca(\varphi(p_1),\varphi(p_2)) = \varphi(p)$, and $\varphi(p_1)\not\sim\varphi(p_2)$,
\item $E\left(\varphi(p),\varphi(p_1),\varphi(p_2)\right)=\mathbb{D},$ that is, $\varphi(p_1)$ and $\varphi(p_2)$ are descendants of $\varphi(p)$, and the previous two cases do not apply.
\end{enumerate}
\end{enumerate}
\end{definition}





An edge \( (u,v) \) of \( P \) is called a \emph{switching edge} (w.r.t.\ a reconciliation \( \varphi \)) if parasite \( v \) switches host (i.e., \( \varphi(v)\not\sim \varphi(u) \); see \Cref{fig:HPsigma eg} for examples). 

The function $E_\varphi$ partitions the set of internal parasite nodes into three disjoint subsets according to their event; these subsets are denoted by $V^\mathbb{C}(P), \,V^\mathbb{D}(P),\,V^\mathbb{S}(P)$. 
In a reconciliation, an internal parasite node may also be associated with one or more \emph{loss events}; see \cite{capybara,Donati2015,Tofigh11} for details. 

Given a cost vector $\vec{c}$ that assigns a value to each event type, the cost of a reconciliation is the sum of the costs of its events. 
\begin{definition}
The \emph{reconciliation problem} is: given two phylogenetic trees $H$ and $P$, a function $\sigma: L(P)\to L(H),$ and a cost vector $\vec{c}$, find a reconciliation $\varphi$ of $(H, P, \sigma)$ of minimum cost.
\end{definition}
From hereon, we refer to instances of the reconciliation problem as a 3-tuple \( (H,P,\sigma) \) without specifying the cost vector since the costs are usually chosen based on the specific pairs of host and parasite trees. 
Usually, cospeciation is assumed to be cost-free, and host-switch is assigned a higher cost than the other events. 

To be clear, the phylogenetic reconciliation problem is used for other cophylogenetic relationships besides host-symbiont co-evolution. 
This includes not only interaction between genes and species, but also other areas such as biogeography (see \cite[Figure~2]{Menet2022} for a quick view). 
Throughout this paper, we adhere to the language of host-parasite co-evolution for simplicity.

\begin{definition}
For a fixed integer \( d\geq 2 \), the problem \emph{distance-bounded reconciliation\( (d) \)} is: given two phylogenetic trees $H$ and $P$, a function $\sigma: L(P)\to L(H),$ and a cost vector $\vec{c}$, find a reconciliation $\varphi$ of $(H, P, \sigma)$ of minimum cost such that $d_H(\varphi(u),\varphi(v))\leq d$ for every switching edge $(u,v)$.
\end{definition}




Since we exclusively deal with directed graphs, path (resp.\ cycle) in this paper refers to directed path (resp.\ directed cycle). 
To be clear, distance between two nodes \( x \) and \( y \) in a (directed) graph \( H \) (denoted by \( d_H(x,y) \)) refers to the undirected distance throughout this paper. 
We write a path in a graph as a comma-separated list of vertices and/or subpaths in the order it appears in the path. 
For example, consider a path \( P \) with vertex set \( \{u,v,w,x,y,z\} \) and arc set \( \{(u,v),(v,w),(w,x),(x,y),(y,z)\} \), and the subpath \( Q \) of \( P \) with vertex set \( \{v,w,x\} \) and arc set \( \{(v,w),(w,x)\} \). 
We may write \( Q \) as \( v,w,x \), and we may write \( P \) as \( u,v,w,x,y,z \) or \( u,v;Q;x,y,z \). 
A path segment in a cycle could also be written in a similar way. 

\section{Acyclicity and strong acyclicity of reconciliation}\label{sec:acyclicity}
Phylogenetic reconciliation is sometimes studied with host and/or parasite tree being dated. 
A dated rooted tree is a rooted tree together with a function \( f \) that maps each node of the tree to a real interval representing a time interval such that \( f(u) \) strictly precedes \( f(v) \) whenever \( u \) is the parent of \( v \) in \( T \); such a function \( f \) is referred to as a dating function for the tree.
The dating information is hard to obtain in practice and is unreliable, and hence both trees are typically assumed to be undated. 
Nevertheless, for a reconciliation \( \varphi \) to be biologically feasible, there should exist a dating function \( f \) for the host tree \( H \) such that two nodes \( x \) and \( y \) co-existed in some time interval  (i.e., \( f(x)\cap f(y) \) is non-empty) whenever some parasite did a host switch from \( x \) to \( y \) (or from \( y \) to \( x \)) in the  reconciliation \( \varphi \). 
For certain reconciliations, it is impossible for such a dating function to exist due to the existence of certain sets of host switches in the reconciliation. 
Such reconciliations are called time-inconsistent or cyclic. 
Other reconciliations are called time-consistent or acyclic. 
The acyclicity of a reconciliation \( \varphi \) is formally defined in terms of an associated directed graph \( D^\varphi \) \cite{stolzer12,capybara,Calamoneri2019}.
\begin{definition}[\cite{stolzer12,Donati2015}]\label{def:acyclicity_condition_std}
For an instance \( (H,P,\sigma) \) of the reconciliation problem and a reconciliation~\( \varphi \), \textbf{\boldmath \( D^\varphi \)} denotes the graph obtained from \( H \) by adding arcs \( (p(\varphi(u),\varphi(u')) \), \( (p(\varphi(u),\varphi(v')) \), \( (p(\varphi(v),\varphi(u')) \) and \( (p(\varphi(v),\varphi(v')) \) for every pair of (not necessarily distinct) switching edges \( (u,v) \) and \( (u',v') \) with \( u \) being an ancestor of \( u' \) in \( P \). 
\end{definition}
\noindent
\emph{Remark:} It is understood that an arc such as \( (p(\varphi(u)),\varphi(u')) \) need not be added if \( p(\varphi(u)) \) is an ancestor of \( \varphi(u') \) (it is redundant as there is a \( p(\varphi(u)),\varphi(u') \)-path already in \( H \)). 
~\\

An example is exhibited in \Cref{fig:Dphi Gphi eg}. 
Note that the above conditions (referred to as tofig-coditions in \cite{Rusin2014}) capture the following: 
(i)~from the ancestry induced by the parasite tree, we know that $u'$ happened strictly after $u$; 
(ii)~the switching edge $(u,v)$ indicates that $\varphi(u)$ and $\varphi(v)$ co-existed in some time interval, and thus 
everything that has happened before $\varphi(u)$ has happened before $\varphi(v)$ and vice versa (this means that \( \varphi(u) \) and the parent of \( \varphi(v) \) cannot co-exist in any time interval, and \( \varphi(v) \) and the parent of \( \varphi(u) \) cannot co-exist in any time interval). 




\begin{definition}[Acyclicity \cite{stolzer12,Donati2015}]
For an instance \( (H,P,\sigma) \) of the reconciliation problem, a reconciliation \( \varphi \) is said to be \emph{acyclic} \( ( \)i.e., time-consistent\( ) \) if there is no cycle in \( D^\varphi \). 
\end{definition}
The above definition is the reason why time-consistent reconciliations are called acyclic. 
Since it is NP-complete to check for existence of time-consistent reconciliations, many relaxations have been proposed including discretizing time (for potential dating functions on the host tree). 
Incorporating discrete time slices into the model (rather than allowing the time to be continuous) has been studied from at least two angles: (i)~fixing time slices in advance (thereby compromising on the optimality of the solution), and (ii)~using discretization as a very close approximation of the continuous time; a criticism towards the latter is that construction of suitable time slices is challenging \cite[Section~1.5]{Rusin2014}. 
For the approach of fixing time slices in advance, it is proved in \cite{Gorbunov2009} that fixing the time slices in advance solely based on the graph-theoretic distance of the node from the root makes the problem polynomial-time solvable~\cite{Rusin2014}. 
Needless to say, this condition is too strong. 

We introduce a version of discretization for distance-bounded reconciliation which is very likely to be polynomial-time solvable despite not fixing the time slices in advance. 
To this end, we introduce a notion stronger than acyclicity (and weaker than fixing time slices in advance). 
We call it strong acyclicity. 
Strong acyclicity view time as discrete steps (such as \( t=1 \), \( t=2 \), and so on), and host switches between two hosts indicate that they existed at the exact same time rather than just co-existing in some time interval (i.e., when there is a host switch between \( x \) and \( y \), the dating function \( f \) satisfies \( f(x)=f(y) \) rather than just \( f(x)\cap f(y)\neq \emptyset \)).
In other words, strong acyclicity employs in place of co-existence in some time interval, the transitive relation of co-existence.

To be clear, using strong acyclicity in place of acyclicity also compromises on the optimality of the solution. 
Yet, this concession is worthy of consideration if limiting to strongly acyclic reconciliations makes the problem polynomial-time solvable (for example, to serve as an approximation algorithm for the actual reconciliation problem). 
We show that unlike acyclicity, with the strong acyclicity condition, cycles can be checked locally at least for \( d=3 \) and \( d=4 \) (see \Cref{thm:no long without short dist3n4}). 
This can have algorithmic consequences (we discuss this in detail in \Cref{sec:what for algo}). 
Strong acyclicity is also defined in terms of an associated directed graph, which we denote by \( G^\varphi \). 
In \( G^\varphi \), we allow two types of arcs, namely strict arcs and relaxed arcs. 
A \emph{strict} arc \( (u,v) \) denotes that \( u \) pre-dates \( v \) (that is, \( u \) exists before \( v \) and no longer exists when \( v \) comes to existence),  whereas a \emph{relaxed} arc \( (u,v) \) denotes that \( u \) and \( v \) co-existed.
The arcs originally present in the host tree are the strict arcs, and the arcs added to \( H \) are the relaxed arcs. 




\begin{definition}\label{def:acyclicity_condition_ours}
For an instance \( (H,P,\sigma) \) of the reconciliation problem and a reconciliation~\( \varphi \), \textbf{\boldmath \( G^\varphi \)} denotes the graph obtained form \( H \) by adding relaxed arcs \( (\varphi(u),\varphi(v)) \) and \( (\varphi(v),\varphi(u)) \) for every switching edge \( (u,v) \) \( ( \)we take the original arcs in \( H \) to be strict arcs\( ) \). 
\end{definition}

\begin{definition}[Strong acyclicity]
For an instance \( (H,P,\sigma) \) of the reconciliation problem, a reconciliation \( \varphi \) is said to be \emph{strongly acyclic} if there is no cycle in \( G^\varphi \) with one or more strict arcs. 
\end{definition}
An example is exhibited in \Cref{fig:Dphi Gphi eg}. 
~\\

\noindent
\emph{Remark:} Note that arcs in \( G^\varphi \) not in \( A_{par} \) appear as digons (i.e., \( (u,v) \) and \( (v,u) \)). 
Hence, it can also be modeled as a mixed graph, and \( \varphi \) is strongly acyclic if and only if there is no directed mixed cycle in the constructed mixed graph. 
~\\

As mentioned before, phylogenetic reconciliation problem is used not only for host-symbiont interaction, but also gene-species interaction, biogeography, and so forth. 
Phylogenetic reconciliation model is also used outside of biology. For instance, for its role in linguistics, see \cite{Kollner2021} and references therein. 
It is worth mentioning that strong acyclicity is equivalent to acyclicity in applications where discrete time steps is a common assumption. 
Applications of this type are admittedly rare for phylogenetic reconciliation, but could naturally occur in certain fields such as geography (in areas like stratigraphic correlation). 


\begin{figure}[hbtp]
\tikzset{
dot/.style={draw,circle,minimum size=1.5,fill=black,inner sep=0},
tree-edges/.style={dashed, gray},
}
\centering
\begin{subfigure}[c]{0.5\textwidth}
\centering
\begin{tikzpicture}[scale=0.5]
\node[coordinate][label=\( r_H \)] (o){};
\begin{scope}[thick,decoration={
    markings,
    mark=at position 0.6 with {\arrow{>}}}
    ] 
\draw[postaction={decorate}] (o)--+(-3.5,-3.5) node[coordinate][label=above left:\( h_4 \)](j){};
\draw[postaction={decorate}] (o)--+( 2.5,-2.5) node[coordinate][label=above right:\( h_5 \)](k){}; 
\draw[postaction={decorate}] (j)--+(-2,-2) node[coordinate][label=below:\( a\mathstrut \)](j1){};
\draw[postaction={decorate}] (j)--+( 1,-1) node[coordinate][label=above right:\( h_1 \)](j2){};
\draw[postaction={decorate}] (j2)--+(-1,-1) node[coordinate][label=below:\( b\mathstrut \)](j3){};
\draw[postaction={decorate}] (j2)--+( 1,-1) node[coordinate][label=below:\( c\mathstrut \)](j4){};
\draw[postaction={decorate}] (k)--+(-1,-1) node[coordinate][label=above left:\( h_3 \)](k1){};
\draw[postaction={decorate}] (k)--+( 3,-3) node[coordinate][label=below:\( g\mathstrut \)](k2){};
\draw[postaction={decorate}] (k1)--+(-1,-1) node[coordinate][label=above left:\( h_2 \)](k3){};
\draw[postaction={decorate}] (k1)--+( 2,-2) node[coordinate][label=below:\( f\mathstrut \)](k4){};
\draw[postaction={decorate}] (k3)--+(-1,-1) node[coordinate][label=below:\( d \)](k5){};
\draw[postaction={decorate}] (k3)--+( 1,-1) node[coordinate][label=below:\( e\mathstrut \)](k6){};
\end{scope}
\end{tikzpicture}
\caption*{\( H \)}
\end{subfigure}%
\begin{subfigure}[c]{0.5\textwidth}
\centering
\begin{tikzpicture}[scale=0.5]
\node (o)[coordinate][label=\( r_P \)]{};
\begin{scope}[thick,decoration={
    markings,
    mark=at position 0.6 with {\arrow{>}}}
    ] 
\draw[postaction={decorate}] (o)--+(-1,-1) node[coordinate][label=above left:\( u \)](j){};
\draw[postaction={decorate}] (o)--+( 6.5,-6.5) node[coordinate][label=below:\( G \)](k){}; 
\draw[postaction={decorate}] (j)--+(-1.5,-1.5) node[coordinate][label=above left:\( w \)](j1){};
\draw[postaction={decorate}] (j)--+( 4.5,-4.5) node[coordinate][label=above right:\( v \)](j2){};
\draw[postaction={decorate}] (j1)--+(-1,-1) node[coordinate][label=above left:\( x \)](j3){};
\draw[postaction={decorate}] (j1)--+( 4,-4) node[coordinate][label=below:\( A \)](j4){};
\draw[postaction={decorate}] (j2)--+(-1,-1) node[coordinate][label=below:\( G \)](j5){};
\draw[postaction={decorate}] (j2)--+( 1,-1) node[coordinate][label=below:\( F \)](j6){};
\draw[postaction={decorate}] (j3)--+(-1,-1) node[coordinate][label=above left:\( u' \)](j7){};
\draw[postaction={decorate}] (j3)--+( 3,-3) node[coordinate][label=below:\( E \)](j8){};
\draw[postaction={decorate}] (j7)--+(-2,-2) node[coordinate][label=below:\( D \)](j9){};
\draw[postaction={decorate}] (j7)--+( 1,-1) node[coordinate][label=above right:\( v' \)](j10){};
\draw[postaction={decorate}] (j10)--+(-1,-1) node[coordinate][label=below:\( A \)](j11){};
\draw[postaction={decorate}] (j10)--+( 1,-1) node[coordinate][label=below:\( A \)](j12){};
\end{scope}
\end{tikzpicture}
\caption*{\( P \)}
\end{subfigure}%
\caption{An example instance \( (H,P,\sigma) \) of the reconciliation problem (or of the problem distance-bounded reconciliation\( (4) \)). 
Trees \( H \) and \( P \) are shown, and \( \sigma \) maps upper case letters to lower case letters. 
}
\label{fig:HPsigma eg0}
\end{figure}

\begin{figure}[hbtp]
\tikzset{
dot/.style={draw,circle,minimum size=1.5,fill=black,inner sep=0},
tree-edges/.style={dashed, gray},
}
\centering
\begin{subfigure}[c]{0.5\textwidth}
\centering
\begin{tikzpicture}[scale=0.5]
\node (o)[coordinate][label=\( r_P \)]{};
\begin{scope}[thick,decoration={
    markings,
    mark=at position 0.6 with {\arrow{>}}}
    ] 
\draw[postaction={decorate}] (o)--+(-1,-1) node[coordinate][label=above left:\( u \)](j){};
\draw[postaction={decorate}] (o)--+( 6.5,-6.5) node[coordinate][label=below:\( G \)](k){}; 
\draw[postaction={decorate}] (j)--+(-1.5,-1.5) node[coordinate][label=above left:\( w \)](j1){};
\draw[postaction={decorate},draw=gray] (j)--+( 4.5,-4.5) node[coordinate][label=above right:\( v \)](j2){};
\draw[postaction={decorate},draw=gray] (j1)--+(-1,-1) node[coordinate][label=above left:\( x \)](j3){};
\draw[postaction={decorate}] (j1)--+( 4,-4) node[coordinate][label=below:\( A \)](j4){};
\draw[postaction={decorate}] (j2)--+(-1,-1) node[coordinate][label=below:\( G \)](j5){};
\draw[postaction={decorate}] (j2)--+( 1,-1) node[coordinate][label=below:\( F \)](j6){};
\draw[postaction={decorate}] (j3)--+(-1,-1) node[coordinate][label=above left:\( u' \)](j7){};
\draw[postaction={decorate}] (j3)--+( 3,-3) node[coordinate][label=below:\( E \)](j8){};
\draw[postaction={decorate}] (j7)--+(-2,-2) node[coordinate][label=below:\( D \)](j9){};
\draw[postaction={decorate},draw=gray] (j7)--+( 1,-1) node[coordinate][label=above right:\( v' \)](j10){};
\draw[postaction={decorate}] (j10)--+(-1,-1) node[coordinate][label=below:\( A \)](j11){};
\draw[postaction={decorate}] (j10)--+( 1,-1) node[coordinate][label=below:\( A \)](j12){};
\end{scope}
\end{tikzpicture}
\caption*{\( P \)}
\end{subfigure}%
\caption{
An example for switching edges. 
Drawn in gray are the switching edges in \( P \) with respect to the reconciliation \( \varphi \) of the instance \( (H,P,\sigma) \) of \Cref{fig:HPsigma eg0} given by \( \varphi(r_P)=r_H \), \( \varphi(u)=\varphi(v')=\varphi(w)=h_4 \), \( \varphi(v)=\varphi(x)=h_5 \), and \( \varphi(u')=h_2 \). 
}
\label{fig:HPsigma eg}
\end{figure}

\begin{figure}[hbt]
\tikzset{
dot/.style={draw,circle,minimum size=1.5,fill=black,inner sep=0},
tree-edges/.style={dashed, gray},
}
\centering
\begin{subfigure}[c]{0.5\textwidth}
\centering
\begin{tikzpicture}[scale=0.5]
\node[coordinate][label=\( r_H \)] (o){};
\begin{scope}[thick,decoration={
    markings,
    mark=at position 0.6 with {\arrow{>}}}
    ] 
\draw[postaction={decorate}] (o)--+(-3.5,-3.5) node[coordinate][label=above left:\( h_4 \)](j){};
\draw[postaction={decorate}] (o)--+( 2.5,-2.5) node[coordinate][label=above right:\( h_5 \)](k){}; 
\draw[postaction={decorate}] (j)--+(-2,-2) node[coordinate][label=below:\( a\mathstrut \)](j1){};
\draw[postaction={decorate}] (j)--+( 1,-1) node[coordinate][label=right:\( h_1 \)](j2){};
\draw[postaction={decorate}] (j2)--+(-1,-1) node[coordinate][label=below:\( b\mathstrut \)](j3){};
\draw[postaction={decorate}] (j2)--+( 1,-1) node[coordinate][label=below:\( c\mathstrut \)](j4){};
\draw[postaction={decorate}] (k)--+(-1,-1) node[coordinate][label={[label distance=-3pt]above left:\( h_3 \)}](k1){};
\draw[postaction={decorate}] (k)--+( 3,-3) node[coordinate][label=below:\( g\mathstrut \)](k2){};
\draw[postaction={decorate}] (k1)--+(-1,-1) node[coordinate][label={[label distance=-3pt]above left:\( h_2 \)}](k3){};
\draw[postaction={decorate}] (k1)--+( 2,-2) node[coordinate][label=below:\( f\mathstrut \)](k4){};
\draw[postaction={decorate}] (k3)--+(-1,-1) node[coordinate][label=below:\( d \)](k5){};
\draw[postaction={decorate}] (k3)--+( 1,-1) node[coordinate][label=below:\( e\mathstrut \)](k6){};

\end{scope}
\draw[->,gray,thick] (k1)--(j);
\end{tikzpicture}
\caption*{\( D^\varphi \) (without redundant arcs)}
\end{subfigure}%
\begin{subfigure}[c]{0.5\textwidth}
\centering
\begin{tikzpicture}[scale=0.5]
\node[coordinate][label=\( r_H \)] (o){};
\begin{scope}[thick,decoration={
    markings,
    mark=at position 0.6 with {\arrow{>}}}
    ] 
\draw[postaction={decorate}] (o)--+(-5.5,-5.5) node[coordinate][label=above left:\( h_4 \)](j){};
\draw[postaction={decorate}] (o)--+( 4.5,-4.5) node[coordinate][label=above right:\( h_5 \)](k){}; 
\draw[postaction={decorate}] (j)--+(-2,-2) node[coordinate][label=below:\( a\mathstrut \)](j1){};
\end{scope}
\begin{scope}[thick,decoration={
    markings,
    mark=at position 0.7 with {\arrow{>}}}
    ] 
\draw[postaction={decorate}] (j)--+( 1,-1) node[coordinate][label=right:\( h_1 \)](j2){};
\end{scope}
\begin{scope}[thick,decoration={
    markings,
    mark=at position 0.6 with {\arrow{>}}}
    ] 
\draw[postaction={decorate}] (j2)--+(-1,-1) node[coordinate][label=below:\( b\mathstrut \)](j3){};
\draw[postaction={decorate}] (j2)--+( 1,-1) node[coordinate][label=below:\( c\mathstrut \)](j4){};
\draw[postaction={decorate}] (k)--+(-1,-1) node[coordinate][label={[label distance=-3pt]above left:\( h_3 \)}](k1){};
\draw[postaction={decorate}] (k)--+( 3,-3) node[coordinate][label=below:\( g\mathstrut \)](k2){};
\draw[postaction={decorate}] (k1)--+(-1,-1) node[coordinate][label=above left:\( h_2 \)](k3){};
\draw[postaction={decorate}] (k1)--+( 2,-2) node[coordinate][label=below:\( f\mathstrut \)](k4){};
\draw[postaction={decorate}] (k3)--+(-1,-1) node[coordinate][label=below:\( d \)](k5){};
\draw[postaction={decorate}] (k3)--+( 1,-1) node[coordinate][label=below:\( e\mathstrut \)](k6){};
\end{scope}

\draw[gray,->,thick] (j) to[bend left=0] (k);
\draw[gray,->,thick] (k) to[bend right=20] (j);

\draw[gray,->,thick] (j)  to[bend right=0] (k3);
\draw[gray,->,thick] (k3) to[bend left=15] (j);
\end{tikzpicture}
\caption*{\( G^\varphi \)}
\end{subfigure}%
\caption{Directed graphs \( D^\varphi \) and \( G^\varphi \) for the reconciliation \( \varphi \) of \Cref{fig:HPsigma eg}. 
Added arcs are drawn in gray in this diagram for better visibility.}
\label{fig:Dphi Gphi eg}
\end{figure}

\begin{definition}
For a fixed value of \( d \), an instance \( (H,P,\sigma) \) of distance-bounded reconciliation\( (d) \) and a reconciliation \( \varphi \), \textbf{\boldmath \( H^+_d \)} denotes the graph obtained from \( H \) by adding arcs \( (u,v) \) and \( (v,u) \) for every pair of vertices in \( H \) such that \( d_H(u,v)\leq d \) and \( u\not\sim v \) \( ( \)i.e., adding digons for every potential host switch\() \). 
\end{definition}


The strict arcs in our model are precisely parent to child arcs in the host phylogenetic tree 
(\begin{tikzpicture}
\begin{scope}[thick,decoration={
    markings,
    mark=at position 0.75 with {\arrow{>}}}
    ] 
\draw[postaction={decorate}] (0,0)--(0.5,0);
\end{scope}
\end{tikzpicture}
 in diagrams), and we denote the set of such arcs by \( A_{par} \). 


\begin{definition}[Arcs in \( H^+_d \) for \( d\in \{2,3,4\} \)]
~\\
\( A_{par}=\{(h_j,h_k)\colon h_j \text{ is the parent of } h_k (\text{in } H)\} \).\\[5pt]
\( ( \)for \( d\geq 2 \)\( ) \)\\
\( A_{sib}=\{(h_j,h_k),(h_k,h_j)\colon h_j \) and \( h_k \) are siblings\( \} \).\\[5pt]
\( ( \)for \( d\geq 3 \)\( ) \)\\
\( A_{nep}=\{(h_j,h_k),(h_k,h_j)\colon h_j \) is a child of the sibling of \( h_k\} \).\\[5pt]
\( ( \)for \( d\geq 4 \)\( ) \)\\
\( A_{gnep}=\{(h_j,h_k),(h_k,h_j)\colon h_j \) is a grandchild of the sibling of \( h_k\} \).\\
\( A_{cos}=\{(h_j,h_k),(h_k,h_j)\colon \) the parent of \( h_j \) and the parent of \( h_k \) are siblings\( \}\). 
\end{definition}
\noindent
Note: `\( nep \)' stands for nepos/neptis, meaning nephew/niece; and `\( cos \)' stands for cousin. 

\begin{lemma}\label{lem:acyclicity models comparison}
For an instance \( (H,P,\sigma) \) of the reconciliation problem, if a reconciliation~\( \varphi \) is strongly acyclic, then \( \varphi \) is acyclic. 
\end{lemma}
\begin{proof}
To prove the lemma, it suffices to show that if there exists a cycle in \( D^\varphi \), then there is a cycle in \( G^\varphi \) with at least one strict arc. 
Observe that for every \( x,y \)-path in \( P \), there will be a \( \varphi(x),\varphi(y) \)-path in \( G^\varphi \) (for each arc \( (a,b) \) along this path, if it is a switching edge, use the arc \( (\varphi(a),\varphi(b)) \); otherwise, use the \( \varphi(a),\varphi(b) \)-path in \( H \)). 
Let \( (u,v) \) and \( (u',v') \) be switching edges (not necessarily distinct), and let \( u \) be an ancestor of \( u' \). 
Then, there is a \( \varphi(u),\varphi(u') \)-path in \( G^\varphi \), and consequently there is also a \( \varphi(v),\varphi(u') \)-path, \( \varphi(u),\varphi(v') \)-path, and \( \varphi(v),\varphi(v') \)-path in \( G^\varphi \) (recall that \( (\varphi(u),\varphi(v)) \) and \( (\varphi(u'),\varphi(v')) \) are arcs in \( G^\varphi \)). 
As a result, adding the (strict) arc \( (P(\varphi(u)),\varphi(u)) \) or arc \( (P(\varphi(v)),\varphi(v)) \)) to these paths give a \( (P(\varphi(u)),\varphi(u')) \)-path, \( (P(\varphi(u),\varphi(v')) \)-path,\( (P(\varphi(v),\varphi(u')) \)-path, 
\( (P(\varphi(v),\varphi(v')) \)-path in \( G^\varphi \), each with at least one strict arc (provided \( P(\varphi(u)) \) and \( P(\varphi(v)) \) exist). 
Consider a cycle \( C \) in \( D^\varphi \). 
Replacing arcs \( (x,y) \) in \( C \) not from \( H \) by those corresponding \( x,y \)-paths in \( G^\varphi \) produces a directed closed walk in \( G^\varphi \) containing at least one strict arc. 
Therefore, there is a cycle in \( G^\varphi \) containing at least one strict arc. 
\end{proof}

The converse is not true. 
It is easy to observe that reconciliations where every host switching is from a node \( h \) of \( H \) to a descendant of the sibling of \( h \) (called constrained transfer in \cite{Tavernelli22}) are acyclic, but not necessarily strongly acyclic. 
There are also other types of reconciliations which are acyclic but not strongly acyclic. 
For instance, the reconciliation \( \varphi \) of \Cref{fig:HPsigma eg} is an example because \( D^\varphi \) is a Directed Acyclic Graph (abbr.\ DAG), whereas \( G^\varphi \) contains a cycle with strict arcs (see \Cref{fig:Dphi Gphi eg}). 

For distance-bounded reconciliation with distance bound 2, every reconciliation is acyclic~\cite{Tavernelli22}. 
As a result, the problem distance-bounded reconciliation\( (2) \) is polynomial-time solvable. 
In fact, for distance-bounded reconciliation\( (2) \), every reconciliation \( \varphi \) is strongly acyclic, since the only cycles in \( G^\varphi \) are 2-cycles between siblings (they contain only relaxed arcs). 

Recall that the arcs present in \( H \) are denoted by \( A_{par} \), and those are the only strict arcs in \( H^+_d \) and \( G^\varphi \). 
The set of arcs in \( H^+_3 \) is \( A_{par}\cup A_{sib}\cup A_{nep} \). 
The set of arcs in \( H^+_4 \) is \( A_{par}\cup A_{sib}\cup A_{nep}\cup A_{gnep}\cup A_{cos} \). 
Observe that for distance-bounded reconciliation\( (d) \), for any reconciliation \( \varphi \), the graph \( G^\varphi \) is a subgraph of \(  H^+_d \) such that (i)~\( V(G^\varphi)=V(H^+_d) \), (ii)~\( A_{par}\subseteq E(G^\varphi) \), and (iii)~\( (v,u)\in E(G^\varphi) \) for every \( (u,v)\in E(G^\varphi)\setminus A_{par} \). 


We prove the following theorem with the help of two lemmas presented in upcoming sections. 

\begin{theorem}\label{thm:no long without short dist3n4}
Let \( d=3 \) \( ( \)resp.\ \( d=4 \)\( ) \), and let \( (H,P,\sigma) \) an instance of distance-bounded reconciliation\( (d) \). 
A reconciliation \( \varphi \) of the instance \( (H,P,\sigma) \) is strongly acyclic if and only if \( G^\varphi \) does not contain any cycle of \Crefrange{itm:tri1}{itm:tetra2} displayed in \Cref{fig:typeset 1} \( ( \)resp.\ \Crefrange{itm:tri1}{itm:nona9} displayed in \Crefrange{fig:typeset 1}{fig:typeset 9}\( ) \). 
\end{theorem}

\section[The case $d=3$]{\boldmath The case $d=3$}\label{sec:d=3}

\begin{lemma}\label{lem:no long without short dist3}
Let \( H \) be a phylogenetic tree, and let \( G \) be a spanning subgraph of \( H^+_3 \) such that (i)~\( A_{par}\subseteq E(G) \), and (ii)~\( (v,u)\in E(G) \) for every \( (u,v)\in E(G)\setminus A_{par} \). 
Then, there is no (directed) cycle in \( G \) containing at least one arc from \( A_{par} \) if and only if \( G \) does not contain any cycle of the following forms: 
\begin{enumerate}[\textnormal{Type }1.]
\item \( (h_j,h_k,h_{k_1}) \), where \( h_j \) and \( h_k \) are siblings, and \( h_{k_1} \) is the child of \( h_k \); 
\item \( (h_j,h_{j_1},h_k,h_{k_1}) \), where \( h_j \) and \( h_k \) are siblings, \( h_{j_1} \) is a child of \( h_j \), and \( h_{k_1} \) is a child of \( h_k \); and 
\item \( (h_j,h_{k_1},h_{k_3},h_{k_2}) \), where \( h_j \) and \( h_k \) are siblings, the children of \( h_k \) are \( h_{k_1} \) and \( h_{k_2} \), and \( h_{k_3} \) is a child of \( h_{k_1} \).
\end{enumerate}
\end{lemma}
\noindent 
\emph{Remark.} See \Cref{fig:typeset 1} for visualization of cycles of \Crefrange{itm:tri1}{itm:tetra2}. 
\begin{proof}
Assume the contrary. 
Let \( C \) be a cycle of minimum length in \( G \) containing at least one arc from \( A_{par} \). 
Let \( L \) be the set of vertices in \( C \) at the highest level in the rooted tree \( H \). 
Consider \( C[L] \), the subgraph of \( C \) induced by \( L \). 

Note that in \( G \), an arc  between two vertices on the same level must be from \( A_{sib} \), and an arc to a vertex from a vertex at a lower level must be from \( A_{nep} \). 
Recall that \( S_v \) denotes the vertex set of the subtree of \( H \) rooted at \( v \), for each vertex \( v \) of \( H \). 
In particular, all arcs in \( C[L] \) are from \( A_{sib} \). 
Note that if \( C[L] \) contains a cycle \( C^* \), then \( C^* \) must be the cycle \( C \) itself (since \( C^* \) is a subgraph of \( C \), and \( C^* \) cannot be a shorter cycle); this is a contradiction since \( C[L] \) does not contain any arc from \( A_{par} \) unlike \( C \). 
Hence, \( C[L] \) does not contain any cycle (i.e., \( C[L] \) is a DAG). 
Let \( u \) be a source vertex in \( C[L] \). 

Let \( w \) be the start vertex of the arc to \( u \) in \( C \). 
Clearly, \( w \) cannot be at a level higher than \( u \) (by the definition of \( L \)), nor at the level of \( u \) (by the choice of \( u \)). 
Thus, \( w \) lies at level lower than \( u \). 
Hence, \( (w,u)\in A_{nep} \), which means that \( w \) is the child of the sibling of \( u \) (in \( H \)). 
Let \( v \) be the sibling of \( u \). 

Since \( (w,u) \) is an arc in \( C \) from \( S_v \) to \( S_u \), and there is no vertex in \( C \) at a level higher than \( u \) and \( v \), there is an arc \( e \) from \( S_u \) to \( S_v \) in \( C \). 
The arc \( e \) cannot be from \( u \) to \( v \) (otherwise, \( (u,v,w) \) is a Type 1 cycle in \( G \)), nor from a child \( u_1 \) of \( u \) to \( v \) (otherwise, \( (u,u_1,v,w) \) is a Type 2 cycle in \( G \)). 
Hence, \( e \) must be from \( u \) to a child \( x \) of \( v \). 
If \( x=w \), then \( C \) must the 2-cycle \( (u,w) \) and thus \( C \) does not contain any arc from \( A_{par} \); a contradiction. 
Hence, \( x\neq w \). 
Since \( (w,u) \) is an arc in \( C \) from \( S_w \) to \( V(H)\setminus S_w \), there is an arc \( e' \) in \( C \) from \( V(H)\setminus S_w \) to~\( S_w \). 

~\\
\noindent
Case~1: \( e'\in A_{par} \).\\ 
By the definition of \( L \), the start vertex of \( e' \) is \( v \), and thus \( e'=(v,w) \). 
Replacing the path \( v,w,u,x \) in \( C \) by the arc \( (v,x) \) produces a cycle in \( G \) of length smaller than \( C \) and contains an arc from \( A_{par} \), a contradiction (to the choice of \( C \)).  

~\\
\noindent
Case~2: \( e'\in A_{sib} \).\\ 
Then, \( e'=(x,w) \). 
Hence, \( C \) must be the cycle \( (u,x,w) \), and thus does not contain any arc from \( A_{par} \); a contradiction. 

~\\
\noindent
Case~3: \( e'\in A_{nep} \).\\ 
In this case, either (i)~\( e' \) is from a child \( y \) of \( x \) to \( w \), or (ii)~\( e' \) is from \( x \) to a child \( z \) of \( w \). 
In Subcase~(i), \( (u,x,y,w) \) is a Type 3 cycle in \( G \), a contradiction. 
In Subcase~(ii), replacing the path \( w,u,x,z \) in \( C \) by the arc \( (w,z) \) results in a cycle in \( G \) of length smaller than \( C \), a contradiction. 

Since all three cases lead to contradictions, the proof is complete. 
\end{proof}

\section[The case $d=4$]{\boldmath The case $d=4$}\label{sec:d4}
For the case \( d=4 \), we present a result similar to \Cref{lem:no long without short dist3}, and it involves 110 short cycles. 
These short cycles are exhibited in \Crefrange{fig:typeset 1}{fig:typeset 9}. 
In those figures, the arcs with arrow tip in the middle are parent-to-child arcs (i.e., arcs from \( A_{par} \)), whereas the other arcs have arrow tips at the end; 
also, the gray dashed lines are the edges in \( H \) that are not part of the short cycle shown. 
Verbal descriptions of the cycles are provided in \Cref{sec:short cycle list in words}. 

~\\
\emph{Note for online readers:} For different types of cycles, links are provided for easy reading. 
For example, `Type 1' links to the subfigure for Type 1 cycle, whereas the `Type 1' in the subfigure caption links to the verbal description of Type 1 cycle in the appendix. 

\begin{figure}[hbtp]
\tikzset{
dot/.style={draw,circle,minimum size=1.5,fill=black,inner sep=0},
tree-edges/.style={dashed, gray},
}
\centering
\begin{subfigure}[c]{0.25\textwidth}
\centering

\caption*{\Cref{itm:tri1}}
\end{subfigure}%
\begin{subfigure}[c]{0.25\textwidth}
\centering
%
\caption*{\Cref{itm:tetra1}}
\end{subfigure}%
\begin{subfigure}[c]{0.25\textwidth}
\centering
%
\caption*{\Cref{itm:tetra2}}
\end{subfigure}%
\begin{subfigure}[c]{0.25\textwidth}
\centering
%
\caption*{\Cref{itm:tetra3}}
\end{subfigure}%

\vspace{5pt}

\begin{subfigure}[c]{0.25\textwidth}
\centering
%
\caption*{\Cref{itm:penta1}}
\end{subfigure}%
\begin{subfigure}[c]{0.25\textwidth}
\centering
%
\caption*{\Cref{itm:hexa1}}
\end{subfigure}%
\begin{subfigure}[c]{0.25\textwidth}
\centering
%
\caption*{\Cref{itm:tetra4}}
\end{subfigure}%
\begin{subfigure}[c]{0.25\textwidth}
\centering
%
\caption*{\Cref{itm:penta1.5}}
\end{subfigure}%

\begin{subfigure}[c]{0.26\textwidth}
\centering
%
\caption*{\Cref{itm:penta2}}
\end{subfigure}%
\begin{subfigure}[c]{0.26\textwidth}
\centering
%
\caption*{\Cref{itm:hexa2}}
\end{subfigure}%
\begin{subfigure}[c]{0.26\textwidth}
\centering
%
\caption*{\Cref{itm:hexa3}}
\end{subfigure}%
\begin{subfigure}[c]{0.26\textwidth}
\centering
%
\caption*{\Cref{itm:hepta1}}
\end{subfigure}%

\vspace{5pt}

\begin{subfigure}[c]{0.26\textwidth}
\centering
%
\caption*{\Cref{itm:hepta2}}
\end{subfigure}%
\begin{subfigure}[c]{0.26\textwidth}
\centering
%
\caption*{\Cref{itm:octa1}}
\end{subfigure}%
\begin{subfigure}[c]{0.2\textwidth}
\centering
%
\caption*{\Cref{itm:tri2}}
\end{subfigure}%
\begin{subfigure}[c]{0.2\textwidth}
\centering
%
\caption*{\Cref{itm:tri3}}
\end{subfigure}%

\vspace{5pt}

\begin{subfigure}[c]{0.2\textwidth}
\centering
%
\caption*{\Cref{itm:tri4}}
\end{subfigure}%
\begin{subfigure}[c]{0.2\textwidth}
\centering
%
\caption*{\Cref{itm:penta5}}
\end{subfigure}%
\caption{}
\label{fig:typeset 1}
\end{figure}

\begin{figure}[hbtp]
\tikzset{
dot/.style={draw,circle,minimum size=1.5,fill=black,inner sep=0},
tree-edges/.style={dashed, gray},
}
\centering
\begin{subfigure}[c]{0.25\textwidth}
\centering
%
\caption*{\Cref{itm:penta6}}
\end{subfigure}%
\begin{subfigure}[c]{0.25\textwidth}
\centering
%
\caption*{\Cref{itm:penta3}}
\end{subfigure}%

\vspace{5pt}

\begin{subfigure}[c]{0.265\textwidth}
\centering
%
\caption*{\Cref{itm:hexa4}}
\end{subfigure}%
\caption*{\Cref{itm:hepta3}}
\end{subfigure}%
\begin{subfigure}[c]{0.2\textwidth}
\centering
%
\caption*{\Cref{itm:tetra8}}
\end{subfigure}%

\vspace{5pt}

\begin{subfigure}[c]{0.25\textwidth}
\centering
%
\caption*{\Cref{itm:penta10}}
\end{subfigure}%

\vspace{5pt}

\begin{subfigure}[c]{0.25\textwidth}
\centering
%
\caption*{\Cref{itm:hexa11}}
\end{subfigure}%

\vspace{5pt}

\begin{subfigure}[c]{0.25\textwidth}
\centering
%
\caption*{\Cref{itm:hepta5}}
\end{subfigure}%
\begin{subfigure}[c]{0.25\textwidth}
\centering
%
\caption*{\Cref{itm:hepta7}}
\end{subfigure}%
\caption{}
\label{fig:typeset 2}
\end{figure}

\begin{figure}[hbtp]
\tikzset{
dot/.style={draw,circle,minimum size=1.5,fill=black,inner sep=0},
tree-edges/.style={dashed, gray},
}
\centering
\begin{subfigure}[c]{0.25\textwidth}
\centering
%
\caption*{\Cref{itm:octa2}}
\end{subfigure}%
\begin{subfigure}[c]{0.25\textwidth}
\centering
%
\caption*{\Cref{itm:octa3}}
\end{subfigure}%

\begin{subfigure}[c]{0.33\textwidth}
\centering
%
\caption*{\Cref{itm:nona1}}
\end{subfigure}%
\caption{}
\label{fig:typeset 3}
\end{figure}

\begin{figure}[hbtp]
\tikzset{
dot/.style={draw,circle,minimum size=1.5,fill=black,inner sep=0},
tree-edges/.style={dashed, gray},
}
\centering
\begin{subfigure}[c]{0.25\textwidth}
\hspace{5pt}
\end{subfigure}%
\begin{subfigure}[c]{0.33\textwidth}
\centering
%
\caption*{\Cref{itm:hepta10}}
\end{subfigure}%

\vspace{2pt}

\begin{subfigure}[c]{0.33\textwidth}
\centering
%
\caption*{\Cref{itm:hexa17}}
\end{subfigure}%
\begin{subfigure}[c]{0.33\textwidth}
%
\caption*{\Cref{itm:hepta12}}
\end{subfigure}%

\begin{subfigure}[c]{0.3\textwidth}
\hspace{5pt}
\end{subfigure}%
\begin{subfigure}[c]{0.33\textwidth}
\centering
%
\caption*{\Cref{itm:octa6}}
\end{subfigure}%
\caption{}
\label{fig:typeset 4}
\end{figure}

\begin{figure}[hbtp]
\tikzset{
dot/.style={draw,circle,minimum size=1.5,fill=black,inner sep=0},
tree-edges/.style={dashed, gray},
}
\centering
\begin{subfigure}[c]{0.33\textwidth}
\centering
%
\caption*{\Cref{itm:nona2}}
\end{subfigure}%

\begin{subfigure}[c]{0.33\textwidth}
\centering
%
\caption*{\Cref{itm:nona3}}
\end{subfigure}%

\begin{subfigure}[c]{0.33\textwidth}
\centering
%
\caption*{\Cref{itm:nona4}}
\end{subfigure}%
\begin{subfigure}[c]{0.33\textwidth}
\centering
%
\caption*{\Cref{itm:deca1}}
\end{subfigure}%
\caption{}
\label{fig:typeset 5}
\end{figure}

\begin{figure}[hbtp]
\tikzset{
dot/.style={draw,circle,minimum size=1.5,fill=black,inner sep=0},
tree-edges/.style={dashed, gray},
}
\centering
\begin{subfigure}[c]{0.33\textwidth}
\centering
%
\caption*{\Cref{itm:deca2}}
\end{subfigure}%
\caption{}
\label{fig:typeset 6}
\end{figure}

\begin{figure}[hbtp]
\tikzset{
dot/.style={draw,circle,minimum size=1.5,fill=black,inner sep=0},
tree-edges/.style={dashed, gray},
}
\centering
\begin{subfigure}[c]{0.33\textwidth}
\centering
%
\caption*{\Cref{itm:deca3}}
\end{subfigure}%
\begin{subfigure}[c]{0.33\textwidth}
\centering
%
\caption*{\Cref{itm:deca4}}
\end{subfigure}%
\begin{subfigure}[c]{0.33\textwidth}
\centering
%
\caption*{\Cref{itm:undeca1}}
\end{subfigure}%

\begin{subfigure}[c]{0.2\textwidth}
\centering
%
\caption*{\Cref{itm:penta15}}
\end{subfigure}%
\caption{}
\label{fig:typeset 7}
\end{figure}

\begin{figure}[hbtp]
\tikzset{
dot/.style={draw,circle,minimum size=1.5,fill=black,inner sep=0},
tree-edges/.style={dashed, gray},
}
\centering
\begin{subfigure}[c]{0.33\textwidth}
\centering
%
\caption*{\Cref{itm:hexa24}}
\end{subfigure}%
\caption{}
\label{fig:typeset 8}
\end{figure}

\begin{figure}[hbtp]
\tikzset{
dot/.style={draw,circle,minimum size=1.5,fill=black,inner sep=0},
tree-edges/.style={dashed, gray},
}
\centering
\begin{subfigure}[c]{0.3\textwidth}
\hspace{5pt}
\end{subfigure}%
\begin{subfigure}[c]{0.33\textwidth}
\centering
%
\caption*{\Cref{itm:hepta21}}
\end{subfigure}%

\vspace{2pt}

\begin{subfigure}[c]{0.33\textwidth}
\centering
%
\caption*{\Cref{itm:octa15}}
\end{subfigure}%

\vspace{2pt}

\begin{subfigure}[c]{0.33\textwidth}
\centering
%
\caption*{\Cref{itm:nona9}}
\end{subfigure}%
\caption{}
\label{fig:typeset 9}
\end{figure}


\begin{lemma}\label{lem:no long without short dist4}
Let \( H \) be a phylogenetic tree, and let \( G \) be a spanning subgraph of \( H^+_4 \) such that (i)~\( A_{par}\subseteq E(G) \), and (ii)~\( (v,u)\in E(G) \) for every \( (u,v)\in E(G)\setminus A_{par} \). 
Then, there is no (directed) cycle in \( G \) containing at least one arc from \( A_{par} \) if and only if \( G \) does not contain any cycle of \Crefrange{itm:tri1}{itm:nona9} displayed in \Crefrange{fig:typeset 1}{fig:typeset 9}. 
\end{lemma}
\begin{proof}
Assume the contrary. 
Let \( C \) be a shortest cycle in \( G \) (i.e., a cycle of minimum length) containing (at least one) arc from \( A_{par} \). 

\begin{myclaim}\label{clm:no parent without link arc}
If a vertex \( b \) and its parent \( a \) are in \( C \), then the arc \( (a,b) \) is in \( C \). 
\end{myclaim}
\noindent
Suppose that a vertex \( b \) and its parent \( a \) are in \( C \), but \( (a,b) \) is not in \( C \). 
Since \( (a,b) \) is not in \( C \), the distance from \( a \) to \( b \) in \( C \) is at least 2. 
As a result, the \( a,b \)-path in \( C \) can be replaced by the arc \( (a,b) \) to obtain a shorter cycle containing arc from \( A_{par} \). 
This contradiction proves \Cref{clm:no parent without link arc}. 
~\\

Note that in \( G \), an arc between two vertices on the same level must be from \( A_{sib} \) or \( A_{cos} \), and an arc to a vertex from a vertex at a lower level must be from \( A_{nep} \) or \( A_{gnep} \). 

Let \( L \) be the set of vertices in \( C \) at the highest level in the rooted tree \( H \). 
Consider \( C[L] \), the subgraph of \( C \) induced by \( L \). 
Note that if \( C[L] \) contains a cycle \( C^* \), then \( C^* \) must be the cycle \( C \) itself (since \( C^* \) is a subgraph of \( C \), and \( C^* \) cannot be a shorter cycle); this is a contradiction since \( C[L] \) does not contain any arc from \( A_{par} \) unlike \( C \). 
Hence, \( C[L] \) does not contain any cycle (i.e., \( C[L] \) is a DAG). 
Let \( u \) be a source vertex in \( C[L] \). 

Let \( w \) be the start vertex of the arc to \( u \) in \( C \). 
Clearly, \( w \) cannot be at a level higher than \( u \) (by the definition of \( L \)), nor at the level of \( u \) (by the choice of \( u \)). 
Thus, \( w \) lies at level lower than \( u \). 
Hence, \( (w,u)\in A_{nep}\cup A_{gnep} \), which means that \( w \) is a child or grandchild of \( v \), where \( v \) is the sibling of \( u \). 
Let \( P \) be the \( v,w \)-path in \( H \). 
Let \( v_1 \) be the child of \( v \) which is an ancestor of \( w \) (i.e., either \( v_1=w \) or \( P= v,v_1,w \)). 

~\\ 
\begin{myclaim}\label{clm:v not in C}
\( v \) is not in \( C \). 
\end{myclaim}
\noindent
%
On the contrary, assume that \( v \) is in \( C \). 
Consider the arc \( e \) to \( v \) in \( C \). 
By the definition of \( L \), the arc \( e \) is from the same level as \( v \) or from a lower level. 
Hence, one of the following holds: \( e=(u,v) \), arc \( e \) is from a child of \( u \) to \( v \), arc \( e \) is from a grandchild of \( u \) to \( v \), or \( e \) is from a cousin of \( v \) to \( v \).
If \( e=(u,v) \), then \( (u,v;P;w) \) is either a \mytype{itm:tri1} cycle or a \mytype{itm:tetra3} cycle in \( G \) (a contradiction). 
If \( e \) is from a child \( u_1 \) of \( u \) to \( v \), then, \( (u,u_1,v;P;w) \) is either a \mytype{itm:tetra1} cycle or a \mytype{itm:penta1} cycle in \( G\). 
If \( e \) is from a grandchild of \( u \) to \( v \), say \( e=(u_3,v) \) where \( u_1 \) is a child of \( u \) and \( u_3 \) is a child of \( u_1 \), then \( (u,u_1,u_3,v;P;w) \) is either a 
\mytype{itm:penta1} cycle or a \mytype{itm:hexa1} cycle in \( G\).

Hence, to prove that \( v \) is not in \( C \), it suffices to exhibit a contradiction when \( e \) is from a cousin \( t \) of \( v \) to \( v \). 
Let \( s \) be the sibling of \( t \). 
Let the parent of \( v \) together with its descendants form the set \( S_1 \), and let the parent of \( t \) together with its descendants form the set \( S_2 \). 
Since the cycle \( C \) contains an arc from \( S_2 \) to \( S_1 \) (namely \( (t,v) \)) and \( C \) contains only vertices \( u,v,s,t \) and their descendants (by the definition of \( L \)), \( C \) contains an arc \( e' \) from \( S_1 \) to \( S_2 \). 
Since the distance between \( x \) and \( y \) in \( H \) is greater than 4 for all \( x\in S_1\setminus \{u,v\} \) and \( y\in S_2\setminus \{s,t\} \), the arc \( e'\in \{(u,t),(v,t),(u,s),(v,s)\} \). 
Clearly, \( e'\neq (v,t) \) since \( C \) is not the 2-cycle \( (v,t) \). 
Also, \( e'\neq (u,t) \) (otherwise, \( (t,v;P;w,u) \) is a \mytype{itm:tetra4} cycle or a \mytype{itm:penta1.5} cycle in \( G \)). 

We show that \( e'\neq (v,s) \). 
Assume that \( e'=(v,s) \). 
Consider the first arc \( e'' \) in \( C \) from \( S_2 \) to \( S_1 \) when \( C \) is traversed starting from \( s \). 
Similar to \( e' \), evidently \( e''\in \{(t,u),(t,v),(s,u),(s,v)\} \). 
Since \( (t,v) \) is an arc in \( C \), we have \( e''\neq (t,u) \) and \( e''\neq (t,v) \). 
since \( C \) is not the 2-cycle \( (s,v) \), we have \( e''\neq (s,v) \). 
Lastly, \( e''\neq (s,u) \) since \( (w,u) \) is an arc in \( C \). 
Thus, by contradiction, \( e''\neq (v,s) \). 
This proves that \( e'\neq (v,s) \). 

Thus, the only remaining possibility is that \( e'=(u,s) \). 
%
Let \( r \) denote the starting vertex of the arc to \( t \) in \( C \). 
Clearly, \( r \) cannot be the parent of \( t \) (by the definition of \( L \)), \( r\neq u \) (since \( C\) is a cycle), \( r\neq v \) (since \( C\) is a cycle), and \( r\notin S_1\setminus \{u,v\} \) (because \( \text{dist}_H(t,S_1\setminus \{u,v\})>4 \)). 
Therefore, \( r \) is one of the following: (i)~\( s \), (ii)~a child of \( s \), or (iii)~a grandchild of \( s \). 
Let \( Q \) denote the \( s,r \)-path in \( H \). 
Then, cycle \( (u,s;Q;r,t,v;P;w) \) in \( G \) is of (i) either \mytype{itm:penta2} or \mytype{itm:hexa2}, (ii) either \mytype{itm:hexa3} or \mytype{itm:hepta1}, or (iii) either \mytype{itm:hepta2} or \mytype{itm:octa1}. 
Thus, by contradiction, \( e \) is cannot be from a cousin of \( v \) to \( v \) either. 
Since we have a contradiction in all cases, \Cref{clm:v not in C} is proved.

\begin{myclaim}\label{clm:cousin not in C}
No cousin of \( u \) is in \( C \). 
\end{myclaim}
\noindent
If a cousin of \( u \) is in \( C \), then \( C\) must contain an arc from \( \{u,v\} \) to a cousin of \( v \) and an arc from a cousin of \( v \) to \( \{u,v\} \) (for instance, consider arcs in \( C \) between \( S_1 \) and \( S_2 \) where \( S_1 \) and \( S_2 \) are defined as in the proof of \Cref{clm:v not in C}). 
Since \( v \) is not in \( C \) (by \Cref{clm:v not in C}), the latter is an arc in \( C \) from a cousin of \( u \) to \( u \), which is a contradiction since \( (w,u) \) is an arc to \( u \) in \( C \). 
This proves \Cref{clm:cousin not in C}. 
~\\

Recall that \( S_z \) denotes the vertex set of the subtree of \( H \) rooted at \( z \), for each vertex \( z \) of  \( H \). 
Let \( v_2 \) be the sibling of \( v_1 \). 

\begin{myclaim}\label{clm:v1 not in C}
\( v_1 \) is not in \( C \). 
\end{myclaim}
\noindent
\noindent 
On the contrary, assume that \( v_1 \) is in \( C \). 
If \( P= v,v_1,w \), then \( (v_1,w) \) is an arc in \( C \) by \Cref{clm:no parent without link arc}. 
Let \( Q \) denote the \( v_1,u \)-path in \( C \) (i.e., \( v_1,u \) or \( v_1,w,u \)). 
By \Cref{clm:no parent without link arc}, children 
of \( w \) are not in \( C \). 
If \( P= v,v_1,w \), then the sibling of \( w \) is not in \( C \) by \Cref{clm:no parent without link arc}. 

Since \( (w,u) \) is an arc in \( C \) from \( S_v \) to \( V(H)\setminus S_v \), there is an arc \( e \) in \( C \) from \( V(H)\setminus S_v \) to \( S_v \). 
Since \( v \) is not in \( C \), arc \( e \) is one of the following: (i)~\( (u,v_1) \), (ii)~from a child of \( u \) to \( v_1 \), (iii)~\( (u,v_2) \), (iv)~from a child \( u_1 \) of \( u \) to \( v_2 \), (v)~from \( u \) to a child \( v_4 \) of \( v_2 \), or (vi)~from \( u \) to a child \( v_3 \) of \( v_1 \).  
If \( e=(u,v_1) \), then either \( C \) is the 2-cycle \( (u,w) \), or \( (u,v_1,w) \) is a \mytype{itm:tri2} cycle in \( G \). 
This rules out scenario (i). 
If \( e \) is from a child \( u_1 \) of \( u \) to \( v_1 \), then \( (u_1,v_1;Q;u) \) is either a \mytype{itm:tri3} or a \mytype{itm:tetra5} cycle in \( G \), thereby ruling out scenario (ii). 
We deal with scenarios (iii), (iv) and (v) below. 
If \( e \) is from \( u \) to a child \( v_3 \) of \( v_1 \), then \( v_3 \) must be \( w \) (because children of \( w \) are not in \( C \) and in case \( P= v,v_1,w \), the sibling of \( w \) is not in \( C \)), and thus \( C \) is the 2-cycle \( (u,w) \); a contradiction. 
This rules out scenario (vi). 

Note that the arc to \( v_1 \) in \( C \) is one of the following: (a)~\( (u,v_1) \), (b)~from a child of \( u \) to \( v_1 \), (c)~\( (v_2,v_1) \), (d)~from a child \( v_5 \) of \( v_2 \) to \( v_1 \), or (e) from a grandchild of \( v_2 \) to \( v_1 \). 
Scenarios (a) and (b) are already ruled out (see the last paragraph). 
Consider scenario (c). 
The arc to \( v_2 \) in \( C \) cannot be \( v_1 \) (since \( C \) is not the 2-cycle \( (v_1,v_2) \)) or a child of \( v_1 \) (by \Cref{clm:no parent without link arc}), and thus it must be from \( u \) or a child of \( u \). 
If it is from \( u \), then \( (v_1;Q;u,v_2) \) is a \mytype{itm:tri4} or a \mytype{itm:tetra7} cycle in \( G \). 
If it is from a child \( u_1 \) of \( u \), then \( (v_1;Q;u,u_1,v_2) \) is a \mytype{itm:tetra6} or a \mytype{itm:penta5} cycle in \( G \). 
Thus, scenario (c) is ruled out. 

At this point, the scenarios not ruled out are (iii), (iv), (v) and (d), (e). 
In scenario (iv), \( (u,u_1) \) is an arc in \( C \) by \Cref{clm:no parent without link arc}. 

In scenarios (iii) and (iv), there is a \( v_1,v_2 \)-path \( R \) in \( C \) (namely, \( v_1;Q;u,v_2 \) or \( v_1;Q;u,u_1,v_2 \)). 
If \( v_2 \) is in \( C \), then in scenarios (d) and (e), there is a \( v_2,v_1 \)-path \( R' \) in \( G \) (namely, \( v_2,v_5,v_1 \) or \( v_2,v_6,v_7,v_1 \) where \( v_6 \) is a child of \( v_2 \), and \( v_7 \) is a child of \( v_6 \)).  
In scenario (iii), \( R \) and \( R' \) together form a \mytype{itm:tetra2}, \mytype{itm:penta6}, \mytype{itm:penta4}, or \mytype{itm:hexa7} cycle in \( G \). 
This rules out scenario (iii). 
In scenario (iv), \( R \) and \( R' \) together form a cycle from \mytype{itm:penta3} to \mytype{itm:hepta3} in \( G \). 
This rules out scenario (iv). 

At this point, only scenario (v) and scenarios (d), (e) remain. 
Hence, we consider scenario (v) in combination with scenario (d) or (e). 
That is, \( (u,v_4) \) is in \( C \) where \( v_4 \) is a child of \( v_2 \), and there is an arc to \( v_1 \) in \( C \) from \( v_5 \) (in scenario (d)) or from a child \( v_7 \) of \( v_5 \) (in scenario (e)), where \( v_5 \) is a child of \( v_2 \). 
If \( v_4=v_5 \), then either \( C=(v_1,u,v_4) \) (a contradiction since \( C \) contains arc from \( A_{par} \)), or \( G \) contains a \mytype{itm:tetra8}, \mytype{itm:tetra9}, or \mytype{itm:penta8}, cycle. 
By contradiction, we may assume that \( v_4\neq v_5 \). 
Owing to \Cref{clm:no parent without link arc}, \( v_2 \) is not in \( C \). 

Let us focus on the case when \( v_5 \) is in \( C \). 
Consider the arc \( e' \) to \( v_5 \) in \( C \). 
It cannot be from a child of \( v_1 \) (i.e., a cousin of \( v_5 \)) due to \Cref{clm:no parent without link arc}. 
Hence, \( e' \) is from \( v_4 \) or a child/grandchild of \( v_4 \). 
If \( e'=(v_4,v_5) \), then either \( C=(v_1,u,v_4,v_5) \) (a contradiction since \( C \) contains arc from \( A_{par} \)), or \( G \) contains a cycle from \mytype{itm:penta9} to \mytype{itm:hexa9}. 
If \( e' \) is from a child of \( v_4 \), then \( G \) contains a cycle from \mytype{itm:penta11} to \mytype{itm:hepta5}. 
If \( e' \) is from a grandchild of \( v_4 \), then \( G \) contains a cycle from \mytype{itm:hexa12} to \mytype{itm:octa2}. 

It remains to consider the case when \( v_5 \) is not in \( C \). 
In particular, scenario (d) is ruled out. 
Let us call the sibling of \( v_7 \) as \( v_8 \). 
Since \( (v_7,v_1) \) is an arc in \( C \) from \( S_{v_5} \) to \( V(H)\setminus S_{v_5} \), there is an arc \( e'' \) in \( C \) from \( V(H)\setminus S_{v_5} \) to \( S_{v_5} \). 
The arc \( e'' \) cannot be from \( v_4 \) to a child \( v_9 \) of \( v_7 \) (otherwise, the \( v_7,v_9 \)-path in \( C \) can be replaced by the arc \( (v_7,v_9) \) to obtain a shorter cycle containing arc from \( A_{par} \)). 
Hence, the options are as follows, \( e'' \) is:  from \( v_4 \) to \( v_7 \); from a child of \( v_4 \) to \( v_7 \); from a child of \( v_4 \) to \( v_8 \); from \( v_4 \) to \( v_8 \); or from \( v_4 \) to a child of \( v_8 \). 
If \( e''=(v_4,v_7) \), then either \( C=(v_1,u,v_4,v_7) \) (a contradiction since \( C \) contains arc from \( A_{par} \)), or \( G \) contains a \mytype{itm:penta12} cycle. 
If \( e'' \) is from a child of \( v_4 \) to \( v_7 \), then \( G \) contains a \mytype{itm:penta13} or a \mytype{itm:hexa13} cycle. 

Consider the subcase when \( e'' \) is from a child \( v_9 \) of \( v_4 \) to \( v_8 \). 
Then, \( (v_4,v_9) \) is in \( C \) by \Cref{clm:no parent without link arc}. 
Also, the arc to \( v_7 \) in \( C \) cannot be from \( v_9 \) or its sibling, and thus it must be from \( v_8 \) or a child/grandchild of \( v_8 \). 
Therefore, \( G \) contains a cycle from \mytype{itm:hexa14} to \mytype{itm:nona1}, thereby ruling out the subcase. 

Next, consider the other two subcases: i.e., \( e'' \) is from \( v_4 \) to \( v_8 \) or to a child \( v_9 \) of \( v_8 \). 
Due to \Cref{clm:no parent without link arc}, no child of \( v_4 \) is in \( C \). 
Hence, the arc \( e''' \) to \( v_7 \) in \( C \) must be from \( v_8 \) or a child/grandchild of \( v_8 \). 
As a result, if \( e''=(v_4,v_8) \), then either \( C=(v_1,u,v_4,v_8,v_7) \) (a contradiction since \( C \) contains arc from \( A_{par} \)), or 
\( G \) contains a cycle from \mytype{itm:hexa15} to \mytype{itm:octa5}. 
The remaining subcase is \( e''=(v_4,v_9) \) (where \( v_9 \) is a child of \( v_8 \)). 
Let us call the sibling of \( v_9 \) as \( v_{10} \). 
Here, \( e'''\neq (v_8,v_7) \) (otherwise, the \( v_8,v_9 \)-path in \( C \) can be replaced by arc \( (v_8,v_9) \) to obtain a shorter cycle with arc from \( A_{par} \)), and thus  \( e''' \) is \( (v_9,v_7) \), \( (v_{10},v_7) \), \( (v_{11},v_7) \) or \( (v_{12},v_7) \), where \( v_{11} \) is a child of \( v_9 \), and \( v_{12} \) is a child of \( v_{10} \). 
If \( e''' \) is \( (v_9,v_7) \) or \( (v_{11},v_7) \), then either \( C=(v_1,u,v_4,v_9,v_7) \) (a contradiction since \( C \) contains arc from \( A_{par} \)), or \( G \) contains a \mytype{itm:hexa17}, \mytype{itm:hexa18}, or \mytype{itm:hepta12} cycle. 
Therefore, \( e''' \) is either \( (v_{10},v_7) \) or \( (v_{12},v_7) \). 
Either way, no child of \( v_7 \) is in \( C \). 
Hence, if \( v_{10} \) is in \( C \), then the arc to \( v_{10} \) in \( C \) must be from \( v_9 \) or a child/grandchild of \( v_9 \). 
Hence, if \( v_{10} \) is in \( C \), then either \( C=(v_1,u,v_4,v_9,v_{10},v_7) \) (a contradiction since \( C \) contains arc from \( A_{par} \)), or 
\( G \) contains a cycle from \mytype{itm:hepta13} to \mytype{itm:deca1}. 
Suppose that \( v_{10} \) is not in \( C \). 
Then, \( e'''=(v_{12},v_7) \). 
Let us call the sibling of \( v_{12} \) as \( v_{13} \). 
Since \( (v_{12},v_7) \) is an arc in \( C \) from \( S_{v_{10}} \) to \( V(H)\setminus S_{v_{10}} \), there is an arc in \( C \) from \( V(H)\setminus S_{v_{10}} \) to \( S_{v_{10}} \). 
Let \( e^{iv} \) be the first arc in \( C \) from \( V(H)\setminus S_{v_{10}} \) to \( S_{v_{10}} \) when \( C \) is traversed starting from \( v_{12} \). 
Clearly, \( e^{iv} \) is \( (v_9,v_{12}) \), from a child \( v_{11} \) of \( v_9 \) to \( v_{12} \), arc \( (v_9,v_{13}) \), or from a child \( v_{11} \) of \( v_9 \) to \( v_{13} \). 
If \( e^{iv}\in \{(v_9,v_{12}), (v_{11},v_{12})\} \), then either \( C=(v_1,u,v_4,v_9,v_{12},v_7) \) (a contradiction since \( C \) contains arc from \( A_{par} \)), or \( G \) contains a \mytype{itm:hepta16}, \mytype{itm:hepta17}, or \mytype{itm:octa10} cycle. 
Thus, \( e^{iv}\in \{(v_9,v_{13}), (v_{11},v_{13})\} \). 
Either way, the arc to \( v_{12} \) in \( C \) must be from \( v_{13} \) or a child/grandchild of \( v_{13} \). 
Thus, either \( C=(v_1,u,v_4,v_9,v_{13},v_{12},v_7) \) (a contradiction since \( C \) contains arc from \( A_{par} \)), or 
\( G \) contains a cycle from \mytype{itm:octa11} to \mytype{itm:undeca1}. 
By contradiction, \Cref{clm:v1 not in C} is proved. 
~\\

Since \( v_1 \) is not in \( C \), we have \( w\neq v_1 \) and in particular, \( P=v,v_1,w \). 
To recap, \( u \) and \( v \) are siblings, \( v_1 \) and \( v_2 \) are the children of \( v \), node \( w \) is a child of \( v_1 \), arc \( (w,u) \) is in \( C \), vertices \( v,v_1 \), cousins of \( u \) (if any) and the two children of \( w \) are not in \( C \), and no proper ancestor of \( u \) is in \( C \) (by the definition the \( L \) and \( u \)). 
Let \( v_3 \) denote the sibling of \( w \). 
Arc \( (u,w) \) is not in \( C \) since \( C \) is not the 2-cycle \( (u,w) \). 

Since \( (w,u) \) is an arc in \( C \) from \( S_v \) to \( V(H)\setminus S_v \), there is an arc in \( C \) from \( V(H)\setminus S_v \) to \( S_v \). 
Let \( e \) be the last arc from \( V(H)\setminus S_v \) to \( S_v \) when \( C \) is traversed starting from \( u \). 
Clearly, \( e \) has one of the following forms: (i)~\( (u,v_3) \), (ii)~\( (u,v_2) \), (iii)~from a child \( u_1 \) of \( u \) to \( v_2 \) (in this case, \( (u,u_1) \) is an arc in \( C \) by \Cref{clm:no parent without link arc}), (iv)~from \( u \) to a child \( v_4 \) of \( v_2 \).  
In Case~(i) (i.e., \( e=(u,v_3) \)), neither \( v_2 \) nor its children are in \( C \) by the choice of \( e \). 
Hence, in Case~(i), the arc to \( w \) in \( C \) is from \( v_3 \) or from a child/grandchild of \( v_3 \), and thus either \( C=(w,u,v_3) \) (a contradiction since \( C \) contains arc from \( A_{par} \)), or \( G \) contains a \mytype{itm:penta14} or \mytype{itm:hexa19} cycle in \( G \).

Only cases (ii), (iii) and (iv) remain. 
In these cases, \( (u,v_3) \) is not an arc in \( C \). 
Since \( (w,u) \) is an arc in \( C \) from \( S_{v_1} \) to \( V(H)\setminus S_{v_1} \), there is an arc \( e' \) in \( C \) from \( V(H)\setminus S_{v_1} \) to \( S_{v_1} \). 
Since \( (u,v_3) \) is not in \( C \), the arc \( e' \) is (a)~\( (v_2,w) \), (b)~\( (v_5,w) \), (c)~\( (v_2,v_3) \), or \( (v_5,v_3) \), where \( v_5 \) is a child of \( V_2 \). 
Note that \( v_4=v_5 \) and \( v_4\neq v_5 \) are both plausible. 
We consider Case~(iv) later (except for one subcase). 
With the combination of the cases (ii) and (a), written in short as (ii)+(a), \( C=(w,u,v_2) \) (a contradiction since \( C \) contains arc from \( A_{par} \)). 
In (iii)+(a), \( (u,u_1,v_2,w) \) is a \mytype{itm:tetra10} cycle in \( G \). 
In (ii)+(b), \( (w,u,v_2,v_5) \) is a \mytype{itm:tetra11} cycle in \( G \). 
In (iii)+(b), \( (w,u,u_1,v_2,v_5) \) is a \mytype{itm:penta15} cycle in \( G \). 
In the case combinations (ii)+(c), (iii)+(c), (ii)+(d), (iii)+(d) and `(iv)+(d) with \( v_4=v_5 \)', there is a \( w,v_3 \)-path in \( C \) (namely, \( w,u,v_2,v_3 \); \( w,u,u_1,v_2,v_3 \); \( w,u,v_2,v_5,v_3 \); \( w,u,u_1,v_2,v_5,v_3 \); or \( w,u,v_4,v_3 \)), and the arc to \( w \) in \( C \) must be from \( v_3 \) or from a child/grandchild of \( v_3 \). 
In those case combinations, depending on whether the arc to \( w \) in \( C \) is from \( v_3 \), a child of \( v_3 \), or a grandchild of \( v_3 \), either \( C=(w,u,v_2,v_5) \) (a contradiction since \( C \) contains arc from \( A_{par} \)), or 
\( G \) contains a cycle from \mytype{itm:penta16} to \mytype{itm:octa14}. 
This rules out (ii) and (iii). 
In (iv)+(a), the \( v_2,v_4 \)-path in \( C \) can be replaced by the arc \( (v_2,v_4) \) to obtain a shorter cycle containing arc from \( A_{par} \). 
The case combination (iv)+(c) is ruled out by \Cref{clm:no parent without link arc}. 
In the case combination `(iv)+(b) with \( v_4=v_5 \)', the cycle \( C=(w,u,v_4) \); a contradiction since \( C \) contains arc from \( A_{par} \). 
Hence, the remaining subcases are `(iv)+(b) with \( v_4\neq v_5 \)' and `(iv)+(d) with \( v_4\neq v_5 \)'. 

Consider (iv)+(b) with \( v_4\neq v_5 \). 
In this subcase, the arc to \( v_5 \) in \( C \) cannot be from \( v_2,v_1,u \) or \( w \), and thus it must be from \( v_4 \), a child/grandchild of \( v_4 \), or from \( v_3 \). 
If the arc to \( v_5 \) in \( C \) is from \( v_4 \) or a child/grandchild of \( v_4 \), then either \( C=(w,u,v_4,v_5) \) (a contradiction since \( C \) contains arc from \( A_{par} \)), or \( G \) contains a \mytype{itm:penta19} or a \mytype{itm:hexa24} cycle. 
Suppose that the arc to \( v_5 \) in \( C \) is from \( v_3 \). 
If the arc to \( v_3 \) in \( C \) is from a child/grandchild \( x \) of \( w \), then the \( w,x \)-path in \( C \) can be replaced by the \( w,x \)-path in \( H \) to obtain a shorter cycle containing arc from \( A_{par} \). 
Hence, the arc to \( v_3 \) in \( C \) cannot be from \( v_1,w,v_2,u \) or a child/grandchild of \( w \), and thus it must be from \( v_4 \). 
Hence, \( C=(w,u,v_4,v_3,v_5) \), a contradiction since \( C \) contains arc from \( A_{par} \). 

Consider (iv)+(d) with \( v_4\neq v_5 \). 
In this subcase, the arc to \( w \) in \( C \) cannot be from \( u,v_2,v_4 \) or \( v_5 \), and thus it must be from \( v_3 \) or a child/grandchild of \( v_3 \). 
Also, the arc to \( v_5 \) in \( C \) cannot be from \( v_2,v_1,u,v_3 \) or \( w \), and thus it must be from \( v_4 \) or a child/grandchild of \( v_4 \). 
As a result, either \( C=(w,u,v_4,v_5,v_3) \) (a contradiction since \( C \) contains arc from \( A_{par} \)), or 
\( G \) contains a cycle from \mytype{itm:hexa25} to \mytype{itm:nona9}. 

 Since all cases and subcases lead to contradictions, the proof is complete. 
 \end{proof}

\section{Algorithmic consequences}\label{sec:what for algo}
The time-consistency is the hardness core of the reconciliation problem as the problem is polynomial-time solvable in the absence of the time-consistency condition~\cite{Bansal12}. 
There is no known algorithm for distance-bounded reconciliation\( (d) \) (i.e., reconciliation with each transfer distance at most \( d \)). 
This problem is apparently polynomial-time solvable for \( d=2 \) (because every solution is time-consistent when \( d=2 \)). 
One could expect the problem to be polynomial-time solvable for a few more values of \( d \) (such as \( d=3 \) and \( d=4 \)), and NP-complete for sufficiently large \( d \). 

For \( d=2 \), the algorithm (in Section~3.1) of Bansal et al.~\cite{Bansal12} can be used to solve distance-bounded reconciliation\( (d) \). 
For higher values of \( d \), it is necessary to check for time-consistency. 
Note that running an enumeration version of the algorithm in Section~3.3 of \cite{Bansal12} (say, with weight 3 for transfer distances at most \( d \), and \( \infty \) for transfer distances more than \( d \)) is not practical because optimal values of an instance of reconciliation and the same taken as an instance of reconciliation without the time-consistency constraint can differ a lot (thus, even enumerating a few sub-optimal solutions does not suffice). 
One can check whether a solution \( \varphi \) is time-consistent in polynomial time by constructing \( D^\varphi \), and checking whether \( D^\varphi \) contains a cycle. 
But, this method is very inefficient to be employed as a time-consistency-check subroutine in a dynamic programming algorithm, especially because arcs added in \( D^\varphi \) depend on every pair of switching edges irrespective of the graph-theoretic distances between those switching edges in \( H \). 
If it is strong acyclicity in place of acyclicity, the situation is not much different either, even though every arc \( (u,v) \) in \( G^\varphi \) is between pairs of nodes of distance at most \( d \) (i.e., \( d_H(u,v)\leq d \)). 

The major reason why checking time-consistency of full solutions is inefficient for algorithm design is this: such a time-consistency-check subroutine in a dynamic programming algorithm needs to be able to handle solutions on the subtree of \( H \) rooted at the current node of \( H \) visited.  
However, checking for existence of cycles locally could circumvent this issue (whenever a new transfer is considered, it could be checked whether including this transfer introduces a short cycle from one of the cycle types). 
Thus, the ability to check for existence of cycles locally could make it possible to adapt the algorithm of Bansal et al.~\cite{Bansal12} and thereby solve distance-bounded reconciliation\( (d) \) for a few small values of \( d \), such as \( d=3 \) or \( d=4 \). 

Using \Cref{thm:no long without short dist3n4}, we outline an adaptation idea for \( d=3 \): in place of \( c(g,s) \), a list of values \( c(g,s)^{(0)} \), \( c(g,s)^{(1)} \), \dots, \( c(g,s)^{(1024)} \) are computed recursively where \( c(g,s)^{(0)} \) is the minimum cost of distance-bounded reconciliation\( (3) \) between \( G_g \) and \( H_s \) when there is no transfer within distance 3 from \( s \) for the optimal values used at lower nodes, \( c(g,s)^{(1)} \) is the similar cost when the only transfer within distance 3 from \( s \) is from a child of \( s \) to the other child of \( s \), and so on; use \( \min_j c(g,s)^{(j)} \) in place of \( c(g,s) \) in the computation of all variables except \( c_\Theta(g,s) \); and in the computation of \( c_\Theta(g,s) \), in place of \( c(g,s) \), use the minimum of \( c(g,s)^{(j)} \) over values of \( j \) that does not create a Type 1, 2 or 3 cycle when used in combination with the latest transfer (to be clear, this adaptation requires keeping track of the node mappings used for optimal values computed up to 3 levels below the current node, which is non-trivial, but doable). 
But, this algorithmic idea has a serious flaw. 
The algorithm of Bansal et al.~\cite{Bansal12} employs in combination with the dynamic programming paradigm, a local greedy choice (e.g., Step 20 of the algorithm). 
This greedy choice could become erroneous when certain node mappings are forbidden (due to the corresponding partial reconciliations being not strongly acyclic). 
Overcoming this shortcoming via very detailed analyses is needed to apply our result to construct an algorithm. 

\section{Open Problems}\label{sec:open}
The problem distance-bounded reconciliation\( (2) \) is polynomial-time solvable. 
In contrast, the complexity of distance-bounded reconciliation\( (d) \) is open for \( d\geq 3 \). 
\Cref{thm:no long without short dist3n4} proved that for \( d\in \{3,4\} \), forbidding a finite number of short cycles suffices to ensure strong acyclicity of a reconciliation, thereby making it possible to check for cycles locally.  
We conjecture that this holds for \( d\geq 5 \) as well. 
If this conjecture is true, an interesting direction to pursue is whether the process of building a finite list of short cycles (such a list need not be unique) and producing a proof that forbidding them suffices to ensure strong acyclicity can both be automated. 
Fixing the algorithmic idea mentioned at the end of the previous section is another future direction. 


\subsubsection*{Acknowledgements}
A.M.\ has been partially supported by the project “NextGRAAL: Next-generation algorithms for constrained GRAph visuALization” funded by MUR Progetti di Ricerca di Rilevante Interesse Nazionale (PRIN) Bando 2022 - Grant ID 2022ME9Z78. B.S.\ has been partially supported by the project “EXPAND: scalable algorithms for EXPloratory Analyses of heterogeneous and dynamic Networked Data”, funded by MUR Progetti di Ricerca di Rilevante Interesse Nazionale (PRIN) Bando 2022 - Grant ID 2022TS4Y3N.
C.A.\ thanks Sandhya T.P.\ for useful discussions.

\subsubsection*{Disclosure of Interests}
The authors have no competing interests to declare that are relevant to the content of this article. 
%
%
%
\bibliographystyle{splncs04}
\bibliography{references}

@ARTICLE{Tofigh11, 
author={Ali Tofigh and Michael Hallett and Jens Lagergren}, 
journal={IEEE/ACM Transactions on Computational Biology and Bioinformatics}, 
title={Simultaneous Identification of Duplications and Lateral Gene Transfers}, 
year={2011}, 
volume={8}, 
number={2}, 
pages={517-535}, 
doi={10.1109/TCBB.2010.14}, 
}

@article{Bansal12,
author = {Bansal, Mukul S. and Alm, Eric J. and Kellis, Manolis},
title = {Efficient algorithms for the reconciliation problem with gene duplication, horizontal transfer and loss},
journal = {Bioinformatics},
volume = {28},
number = {12},
pages = {i283-i291},
year = {2012},
doi = {10.1093/bioinformatics/bts225},
}

@article{stolzer12,
    author = {Stolzer, Maureen and Lai, Han and Xu, Minli and Sathaye, Deepa and Vernot, Benjamin and Durand, Dannie},
    title = "{Inferring duplications, losses, transfers and incomplete lineage sorting with nonbinary species trees}",
    journal = {Bioinformatics},
    volume = {28},
    number = {18},
    pages = {i409-i415},
    year = {2012},
    month = {09},
    issn = {1367-4803},
    doi = {10.1093/bioinformatics/bts386},
}

@article{capybara,
  author    = {Yishu Wang and
               Arnaud Mary and
               Marie{-}France Sagot and
               Blerina Sinaimeri},
  title     = {Capybara: equivalence ClAss enumeration of coPhylogenY event-BAsed
               ReconciliAtions},
  journal   = {Bioinformatics},
  volume    = {36},
  number    = {14},
  pages     = {4197--4199},
  year      = {2020},
  doi = {10.1093/bioinformatics/btaa498}
 }

@article{Donati2015,
  author    = {Beatrice Donati and
               Christian Baudet and
               Blerina Sinaimeri and
               Pierluigi  Crescenzi and
               Marie{-}France Sagot},
  title     = {{\sc Eucalypt:} Efficient tree reconciliation enumerator},
  journal   = {Algorithms for Molecular Biology},
  volume    = {10},
  number={1},
  pages     = {3},
  year      = {2015},
  doi = {10.1186/s13015-014-0031-3}
}

@Article{Charleston2003,
  author =	 {Michael A. Charleston},
  title =	 {Recent results in cophylogeny mapping},
  journal =	 {Advances in Parasitology},
  year =	 2003,
  volume =	 54,
  pages =	 {303--330},
  month =	 {December},
  doi= {10.1016/s0065-308x(03)54007-6} ,
}

@Article{MM2005,
  author =	 {Daniel Merkle and Martin Middendorf},
  title =	 {Reconstruction of the cophylogenetic history of
                  related phylogenetic trees with divergence timing
                  information},
  journal =	 {Theory in Biosciences},
  year =	 2005,
  volume =	 123,
  pages =	 {277--299},
  doi = {10.1016/j.thbio.2005.01.003},
}

@Article{OFCLH2011,
author = {Ovadia, Yaniv J. and Fielder, Daniel and Conow, Chris and Libeskind-Hadas, Ran},
title = {The Cophylogeny Reconstruction Problem Is NP-Complete},
journal = {Journal of Computational Biology},
volume = {18},
number = {1},
pages = {59-65},
year = {2011},
doi = {10.1089/cmb.2009.0240},
}

@article{Jacox2016,
author = {Jacox, Edwin and Chauve, Cedric and Sz\"{o}ll\H{o}si, Gergely J. and Ponty, Yann and Scornavacca, Celine}, 
title = {{ecceTERA}: Comprehensive gene tree-species tree reconciliation using parsimony},
year = {2016}, 
journal = {Bioinformatics} ,
doi={10.1093/bioinformatics/btw105},
}

@Article{Bansal2013,
    author={Bansal, Mukul S.
    and Alm, Eric J.
    and Kellis, Manolis},
    title={Reconciliation revisited: handling multiple optima when reconciling with duplication, transfer, and loss},
    journal={Journal of computational biology : a journal of computational molecular cell biology},
    year={2013},
    month={Oct},
    edition={2013/09/14},
    publisher={Mary Ann Liebert, Inc.},
    volume={20},
    number={10},
    pages={738-754},
    issn={1557-8666},
    doi={10.1089/cmb.2013.0073},
}

@article{Etherington2006,
   author = "Etherington, Graham J. and Ring, Susan M. and Charleston, Michael A. and Dicks, Jo and Rayward-Smith, Vic J. and Roberts, Ian N.",
   title = "Tracing the origin and co-phylogeny of the caliciviruses", 
   journal= "Journal of General Virology",
   year = "2006",
   volume = "87",
   number = "5",
   pages = "1229-1235",
   publisher = "Microbiology Society",
   doi = {10.1099/vir.0.81635-0},
  }

@article{Lei2014,
    author = {Lei, Bonnie R. AND Olival, Kevin J.},
    journal = {PLOS Neglected Tropical Diseases},
    publisher = {Public Library of Science},
    title = {Contrasting Patterns in Mammal–Bacteria Coevolution: Bartonella and Leptospira in Bats and Rodents},
    year = {2014},
    month = mar,
    volume = {8},
    pages = {1-11},
    number = {3},
    doi = {10.1371/journal.pntd.0002738},
}

@article{PENNINGTON2015,
title = {The Chagas disease domestic transmission cycle in Guatemala: Parasite-vector switches and lack of mitochondrial co-diversification between Triatoma dimidiata and Trypanosoma cruzi subpopulations suggest non-vectorial parasite dispersal across the Motagua valley},
journal = {Acta Tropica},
volume = {151},
pages = {80-87},
year = {2015},
note = {Ecology and diversity of Trypanosoma cruzi},
author = {Pamela M. Pennington and Louisa Alexandra Messenger and Jeffrey Reina and José G. Juárez and Gena G. Lawrence and Ellen M. Dotson and Martin S. Llewellyn and Celia Cordón-Rosales},
doi = {10.1016/j.actatropica.2015.07.014},
}

@Article{CFOLH2010,
  author =	 {C. Conow and D. Fielder and Y. Ovadia and
                  R. Libeskind-Hadas},
  title =	 {Jane: A new tool for the cophylogeny reconstruction problem},
  journal =	 {Algorithms for Molecular Biology},
  year =	 2010,
  volume =	 5,
  number =	 16,
  month =	 {February},
}

@ARTICLE{Tavernelli22,
  author={Tavernelli, Daniele and Calamoneri, Tiziana and Vocca, Paola},
  journal={IEEE/ACM Transactions on Computational Biology and Bioinformatics}, 
  title={Linear Time Reconciliation With Bounded Transfers of Genes}, 
  year={2022},
  volume={19},
  number={2},
  pages={1009-1017},
  doi={10.1109/TCBB.2020.3027207}}

@article {DeVienne2007,
  author =	 {D. M. De Vienne and T. Giraud and J. A. Skyhoff},
  title =	 {{When can host shifts produce congruent host and
                  parasite phylogenies? A simulation approach}},
  journal =	 {Journal of Evolutionary Biology},
  volume =	 20,
  number =	 4,
  pages =	 {1428--1438},
  year =	 2007,
}

@article{Poulin2003,
  author =	 {R. Poulin and D. Mouillot},
  title =	 {{Parasite specialization from a phylogenetic
                  perspective: a new index of host specificity}},
  journal =	 {Parasitology},
  volume =	 126,
  year =	 2003,
  pages =	 {473--480},
  issue =	 5,
}

@article{Calamoneri2019,
  title = {Co-divergence and tree topology},
  volume = {79},
  ISSN = {1432-1416},
  DOI = {10.1007/s00285-019-01385-w},
  number = {3},
  journal = {Journal of Mathematical Biology},
  publisher = {Springer Science and Business Media LLC},
  author = {Calamoneri,  T. and Monti,  A. and Sinaimeri,  B.},
  year = {2019},
  month = jun,
  pages = {1149–1167}
}

@article{Rusin2014,
  title={Reconciliation of gene and species trees},
  author={Rusin, L Yu and Lyubetskaya, EV and Gorbunov, K Yu and Lyubetsky, VA},
  journal={BioMed research international},
  volume={2014},
  number={1},
  pages={642089},
  year={2014},
  publisher={Wiley Online Library},
  doi={10.1155/2014/642089}
}

@article{Gorbunov2009,
  title={Reconstructing the evolution of genes along the species tree},
  author={Gorbunov, K Yu and Lyubetsky, VA},
  journal={Molecular Biology},
  volume={43},
  number={5},
  pages={881--893},
  year={2009},
  publisher={Springer},
  doi={10.1134/S0026893309050197}
}

@phdthesis{Kollner2021,
  title={Automatic loanword identification using tree reconciliation},
  author={K{\"o}llner, Marisa},
  year={2021},
  school={Dissertation, T{\"u}bingen, Universit{\"a}t T{\"u}bingen, 2021}
}

@article{Menet2022,
  title={Phylogenetic reconciliation},
  author={Menet, Hugo and Daubin, Vincent and Tannier, Eric},
  journal={PLOS Computational Biology},
  volume={18},
  number={11},
  pages={e1010621},
  year={2022},
  publisher={Public Library of Science San Francisco, CA USA},
  doi={10.1371/journal.pcbi.1010621}
}
\appendix
\section{The list of short cycles in \Cref{lem:no long without short dist4}}\label{sec:short cycle list in words}
\begin{enumerate}[\textnormal{Type} 1.]
\item \( (h_j,h_k,h_{k_1}) \), where \( h_j \) and \( h_k \) are siblings, and \( h_{k_1} \) is a child of \( h_k \); \label{itm:tri1}
\item \( (h_j,h_{j_1},h_k,h_{k_1}) \), where \( h_j \) and \( h_k \) are siblings, \( h_{j_1} \) is a child of \( h_j \), and \( h_{k_1} \) is a child of \( h_k \); \label{itm:tetra1} 
\item \( (h_j,h_{k_1},h_{k_3},h_{k_2}) \), where \( h_j \) and \( h_k \) are siblings, \( h_{k_1} \) and \( h_{k_2} \) are the children of \( h_k \), and \( h_{k_3} \) is a child of \( h_{k_1} \); \label{itm:tetra2} 
\item \( (h_j,h_k,h_{k_1},h_{k_3}) \), where \( h_j \) and \( h_k \) are siblings, \( h_{k_1} \) is a child of \( h_k \), and \( h_{k_3} \) is a child of \( h_{k_1} \); \label{itm:tetra3} 
\item \( (h_j,h_{j_1},h_k,h_{k_1},h_{k_3}) \), where \( h_j \) and \( h_k \) are siblings, \( h_{j_1} \) is a child of \( h_j \), node \( h_{k_1} \) is a child of \( h_k \), and \( h_{k_3} \) is a child of \( h_{k_1} \); \label{itm:penta1}
\item \( (h_j,h_{j_1},h_{j_3},h_k,h_{k_1},h_{k_3}) \), where \( h_j \) and \( h_k \) are siblings, \( h_{j_1} \) is a child of \( h_j \), node \( h_{j_3} \) is a child of \( h_{j_1} \), node \( h_{k_1} \) is a child of \( h_k \), and \( h_{k_3} \) is a child of \( h_{k_1} \); \label{itm:hexa1}
\item \( (h_j,h_\ell,h_k,h_{k_1}) \), where \( h_j \) and \( h_k \) are siblings, \( h_{k_1} \) is a child of \( h_k \), and \( h_\ell \) is a cousin of \( h_ j \) and \( h_k \); \label{itm:tetra4}
\item \( (h_j,h_\ell,h_k,h_{k_1},h_{k_3}) \), where \( h_j \) and \( h_k \) are siblings, \( h_{k_1} \) is a child of \( h_k \), node \( h_{k_3} \) is a child of \( h_{k_1} \), and \( h_\ell \) is a cousin of \( h_ j \) and \( h_k \); \label{itm:penta1.5}
\item \( (h_j,h_\ell,h_m,h_k,h_{k_1}) \), where \( h_j \) and \( h_k \) are siblings, \( h_{k_1} \) is a child of \( h_k \), node \( h_\ell \) is a cousin of \( h_ j \) and \( h_k \), and \( h_m \) is the sibling of \( h_\ell \); \label{itm:penta2}
\item \( (h_j,h_\ell,h_m,h_k,h_{k_1},h_{k_3}) \), where \( h_j \) and \( h_k \) are siblings, \( h_{k_1} \) is a child of \( h_k \), node \( h_{k_3} \) is a child of \( h_{k_1} \), node \( h_\ell \) is a cousin of \( h_ j \) and \( h_k \), and \( h_m \) is the sibling of \( h_\ell \); \label{itm:hexa2}
\item \( (h_j,h_\ell,h_{\ell_1},h_m,h_k,h_{k_1}) \), where \( h_j \) and \( h_k \) are siblings, \( h_{k_1} \) is a child of \( h_k \), node \( h_\ell \) is a cousin of \( h_j \) and \( h_k \), node \( h_{\ell_1} \) is a child of \( h_\ell \), and \( h_m \) is the sibling of \( h_\ell \); \label{itm:hexa3}
\item \( (h_j,h_\ell,h_{\ell_1},h_m,h_k,h_{k_1},h_{k_3}) \), where \( h_j \) and \( h_k \) are siblings, \( h_{k_1} \) is a child of \( h_k \), node \( h_{k_3} \) is a child of \( h_{k_1} \), node \( h_\ell \) is a cousin of \( h_ j \) and \( h_k \), node \( h_{\ell_1} \) is a child of \( h_\ell \), and \( h_m \) is the sibling of \( h_\ell \); \label{itm:hepta1}
\item \( (h_j,h_\ell,h_{\ell_1},h_{\ell_3},h_m,h_k,h_{k_1}) \), where \( h_j \) and \( h_k \) are siblings, \( h_{k_1} \) is a child of \( h_k \), node \( h_\ell \) is a cousin of \( h_ j \) and \( h_k \), node \( h_{\ell_1} \) is a child of \( h_\ell \), node \( h_{\ell_3} \) is a child of \( h_{\ell_1} \), and \( h_m \) is the sibling of \( h_\ell \); \label{itm:hepta2}
\item \( (h_j,h_\ell,h_{\ell_1},h_{\ell_3},h_m,h_k,h_{k_1},h_{k_3}) \), where \( h_j \) and \( h_k \) are siblings, \( h_{k_1} \) is a child of \( h_k \), node \( h_{k_3} \) is a child of \( h_{k_1} \), node \( h_\ell \) is a cousin of \( h_ j \) and \( h_k \), node \( h_{\ell_1} \) is a child of \( h_\ell \), node \( h_{\ell_3} \) is a child of \( h_{\ell_1} \), and \( h_m \) is the sibling of \( h_\ell \); \label{itm:octa1}
\item \( (h_j,h_{k_1},h_{k_3}) \), where \( h_j \) and \( h_k \) are siblings, \( h_{k_1} \) is a child of \( h_k \), and \( h_{k_3} \) is a child of \( h_{k_1} \); \label{itm:tri2}
\item \( (h_j,h_{j_1},h_{k_1}) \), where \( h_j \) and \( h_k \) are siblings, \( h_{j_1} \) is a child of \( h_j \), and \( h_{k_1} \) is a child of \( h_k \); \label{itm:tri3}
\item \( (h_j,h_{j_1},h_{k_1},h_{k_3}) \), where \( h_j \) and \( h_k \) are siblings, \( h_{j_1} \) is a child of \( h_j \), node \( h_{k_1} \) is a child of \( h_k \), and \( h_{k_3} \) is a child of \( h_{k_1} \); \label{itm:tetra5}
\item \( (h_j,h_{k_1},h_{k_2}) \), where \( h_j \) and \( h_k \) are siblings, and the children of \( h_k \) are \( h_{k_1} \) and \( h_{k_2} \); \label{itm:tri4}
\item \( (h_j,h_{k_1},h_{k_2},h_{k_3}) \), where \( h_j \) and \( h_k \) are siblings, \( h_{k_1} \) and \( h_{k_2} \) are the children of \( h_k \), and \( h_{k_3} \) is a child of \( h_{k_2} \); \label{itm:tetra7}
\item \( (h_j,h_{j_1},h_{k_1},h_{k_2}) \), where \( h_j \) and \( h_k \) are siblings, \( h_{j_1} \) is a child of \( h_j \), and the children of \( h_k \) are \( h_{k_1} \) and \( h_{k_2} \); \label{itm:tetra6}
\item \( (h_j,h_{j_1},h_{k_1},h_{k_2},h_{k_3}) \), where \( h_j \) and \( h_k \) are siblings, \( h_{j_1} \) is a child of \( h_j \), nodes \( h_{k_1} \) and \( h_{k_2} \) are the children of \( h_k \), and \( h_{k_3} \) is a child of \( h_{k_2} \); \label{itm:penta5}
\item \( (h_j,h_{k_1},h_{k_3},h_{k_5},h_{k_2}) \), where \( h_j \) and \( h_k \) are siblings, \( h_{k_1} \) and \( h_{k_2} \) are the children of \( h_k \), node \( h_{k_3} \) is a child of \( h_{k_1} \), and \( h_{k_5} \) is a child of \( h_{k_3} \); \label{itm:penta6}
\item \( (h_j,h_{k_1},h_{k_3},h_{k_2},h_{k_4}) \), where \( h_j \) and \( h_k \) are siblings, \( h_{k_1} \) and \( h_{k_2} \) are the children of \( h_k \), node \( h_{k_3} \) is a child of \( h_{k_1} \), and \( h_{k_4} \) is a child of \( h_{k_2} \); \label{itm:penta4}
\item \( (h_j,h_{k_1},h_{k_3},h_{k_5},h_{k_2},h_{k_4}) \), where \( h_j \) and \( h_k \) are siblings, \( h_{k_1} \) and \( h_{k_2} \) are the children of \( h_k \), node \( h_{k_3} \) is a child of \( h_{k_1} \), node \( h_{k_5} \) is a child of \( h_{k_3} \), and \( h_{k_4} \) is a child of \( h_{k_2} \); 
\label{itm:hexa7}
\item \( (h_j,h_{j_1},h_{k_1},h_{k_3},h_{k_2}) \), where \( h_j \) and \( h_k \) are siblings, \( h_{j_1} \) is a child of \( h_j \), the children of \( h_k \) are \( h_{k_1} \) and \( h_{k_2} \), and \( h_{k_3} \) is a child of \( h_{k_1} \); \label{itm:penta3}
\item \( (h_j,h_{j_1},h_{k_1},h_{k_3},h_{k_5},h_{k_2}) \), where \( h_j \) and \( h_k \) are siblings, \( h_{j_1} \) is a child of \( h_j \), the children of \( h_k \) are \( h_{k_1} \) and \( h_{k_2} \), node \( h_{k_3} \) is a child of \( h_{k_1} \), and \( h_{k_5} \) is a child of \( h_{k_3} \);
\label{itm:hexa5}
\item \( (h_j,h_{j_1},h_{k_1},h_{k_3},h_{k_2},h_{k_4}) \), where \( h_j \) and \( h_k \) are siblings, \( h_{j_1} \) is a child of \( h_j \), nodes \( h_{k_1} \) and \( h_{k_2} \) are the children of \( h_k \), \( h_{k_3} \) is a child of \( h_{k_1} \), and \( h_{k_4} \) is a child of \( h_{k_2} \); \label{itm:hexa4}
\label{itm:hexa8}
\item \( (h_j,h_{j_1},h_{k_1},h_{k_3},h_{k_5},h_{k_2},h_{k_4}) \), where \( h_j \) and \( h_k \) are siblings, \( h_{j_1} \) is a child of \( h_j \), nodes \( h_{k_1} \) and \( h_{k_2} \) are the children of \( h_k \), node \( h_{k_3} \) is a child of \( h_{k_1} \), node \( h_{k_5} \) is a child of \( h_{k_3} \), and \( h_{k_4} \) is a child of \( h_{k_2} \); 
\label{itm:hepta3}
\label{itm:hepta4}
\item \( (h_j,h_{k_4},h_{k_1},h_{k_3}) \), where \( h_j \) and \( h_k \) are siblings, \( h_{k_1} \) and \( h_{k_2} \) are the children of \( h_k \), node \( h_{k_3} \) is a child of \( h_{k_1} \), and \( h_{k_4} \) is a child of \( h_{k_2} \); 
\label{itm:tetra8}
\item \( (h_j,h_{k_4},h_{k_6},h_{k_1}) \), where \( h_j \) and \( h_k \) are siblings, \( h_{k_1} \) and \( h_{k_2} \) are the children of \( h_k \), node \( h_{k_4} \) is a child of \( h_{k_2} \), and \( h_{k_6} \) is a child of \( h_{k_4} \); 
\label{itm:tetra9}
\item \( (h_j,h_{k_4},h_{k_6},h_{k_1},h_{k_3}) \), where \( h_j \) and \( h_k \) are siblings, \( h_{k_1} \) and \( h_{k_2} \) are the children of \( h_k \), node \( h_{k_3} \) is a child of \( h_{k_1} \), node \( h_{k_4} \) is a child of \( h_{k_2} \) and \( h_{k_6} \) is a child of \( h_{k_4} \); 
\label{itm:penta8}
\item \( (h_j,h_{k_5},h_{k_4},h_{k_1},h_{k_3}) \), where \( h_j \) and \( h_k \) are siblings, \( h_{k_1} \) and \( h_{k_2} \) are the children of \( h_k \), nodes \( h_{k_4} \) and \( h_{k_5} \) are the children of \( h_{k_2} \), and \( h_{k_3} \) is a child of \( h_{k_1} \);
\label{itm:penta9}
\item \( (h_j,h_{k_5},h_{k_4},h_{k_6},h_{k_1}) \), where \( h_j \) and \( h_k \) are siblings, \( h_{k_1} \) and \( h_{k_2} \) are the children of \( h_k \), nodes \( h_{k_4} \) and \( h_{k_5} \) are the children of \( h_{k_2} \), and \( h_{k_6} \) is a child of \( h_{k_4} \);
\label{itm:penta10}
\item \( (h_j,h_{k_5},h_{k_4},h_{k_6},h_{k_1},h_{k_3}) \), where \( h_j \) and \( h_k \) are siblings, \( h_{k_1} \) and \( h_{k_2} \) are the children of \( h_k \), nodes \( h_{k_4} \) and \( h_{k_5} \) are the children of \( h_{k_2} \), node \( h_{k_3} \) is a child of \( h_{k_1} \), and \( h_{k_6} \) is a child of \( h_{k_4} \);
\label{itm:hexa9}
\item \( (h_j,h_{k_5},h_{k_7},h_{k_4},h_{k_1}) \), where \( h_j \) and \( h_k \) are siblings, \( h_{k_1} \) and \( h_{k_2} \) are the children of \( h_k \), nodes \( h_{k_4} \) and \( h_{k_5} \) are the children of \( h_{k_2} \), and \( h_{k_7} \) is a child of \( h_{k_5} \);
\label{itm:penta11}
\item \( (h_j,h_{k_5},h_{k_7},h_{k_4},h_{k_6},h_{k_1}) \), where \( h_j \) and \( h_k \) are siblings, \( h_{k_1} \) and \( h_{k_2} \) are the children of \( h_k \), nodes \( h_{k_4} \) and \( h_{k_5} \) are the children of \( h_{k_2} \), node \( h_{k_6} \) is a child of \( h_{k_4} \), and \( h_{k_7} \) is a child of \( h_{k_5} \);
\label{itm:hexa10}
\item \( (h_j,h_{k_5},h_{k_7},h_{k_4},h_{k_1},h_{k_3}) \), where \( h_j \) and \( h_k \) are siblings, \( h_{k_1} \) and \( h_{k_2} \) are the children of \( h_k \), nodes \( h_{k_4} \) and \( h_{k_5} \) are the children of \( h_{k_2} \), node \( h_{k_3} \) is a child of \( h_{k_1} \), and \( h_{k_7} \) is a child of \( h_{k_5} \);
\label{itm:hexa11}
\item \( (h_j,h_{k_5},h_{k_7},h_{k_4},h_{k_6},h_{k_1},h_{k_3}) \), where \( h_j \) and \( h_k \) are siblings, \( h_{k_1} \) and \( h_{k_2} \) are the children of \( h_k \), nodes \( h_{k_4} \) and \( h_{k_5} \) are the children of \( h_{k_2} \), node \( h_{k_3} \) is a child of \( h_{k_1} \), node \( h_{k_6} \) is a child of \( h_{k_4} \), and \( h_{k_7} \) is a child of \( h_{k_5} \);
\label{itm:hepta5}
\item \( (h_j,h_{k_5},h_{k_7},h_{k_9},h_{k_4},h_{k_1}) \), where \( h_j \) and \( h_k \) are siblings, \( h_{k_1} \) and \( h_{k_2} \) are the children of \( h_k \), nodes \( h_{k_4} \) and \( h_{k_5} \) are the children of \( h_{k_2} \), node \( h_{k_7} \) is a child of \( h_{k_5} \), and \( h_{k_9} \) is a child of \( h_{k_7} \);
\label{itm:hexa12}
\item \( (h_j,h_{k_5},h_{k_7},h_{k_9},h_{k_4},h_{k_6},h_{k_1}) \), where \( h_j \) and \( h_k \) are siblings, \( h_{k_1} \) and \( h_{k_2} \) are the children of \( h_k \), nodes \( h_{k_4} \) and \( h_{k_5} \) are the children of \( h_{k_2} \), node \( h_{k_6} \) is a child of \( h_{k_4} \), node \( h_{k_7} \) is a child of \( h_{k_5} \), and \( h_{k_9} \) is a child of \( h_{k_7} \);
\label{itm:hepta6}
\item \( (h_j,h_{k_5},h_{k_7},h_{k_9},h_{k_4},h_{k_1},h_{k_3}) \), where \( h_j \) and \( h_k \) are siblings, \( h_{k_1} \) and \( h_{k_2} \) are the children of \( h_k \), nodes \( h_{k_4} \) and \( h_{k_5} \) are the children of \( h_{k_2} \), node \( h_{k_3} \) is a child of \( h_{k_1} \), node \( h_{k_7} \) is a child of \( h_{k_5} \), and \( h_{k_9} \) is a child of \( h_{k_7} \);
\label{itm:hepta7}
\item \( (h_j,h_{k_5},h_{k_7},h_{k_9},h_{k_4},h_{k_6},h_{k_1},h_{k_3}) \), where \( h_j \) and \( h_k \) are siblings, \( h_{k_1} \) and \( h_{k_2} \) are the children of \( h_k \), nodes \( h_{k_4} \) and \( h_{k_5} \) are the children of \( h_{k_2} \), node \( h_{k_3} \) is a child of \( h_{k_1} \), node \( h_{k_6} \) is a child of \( h_{k_4} \), node \( h_{k_7} \) is a child of \( h_{k_5} \), and \( h_{k_9} \) is a child of \( h_{k_7} \);
\label{itm:octa2}
\item \( (h_j,h_{k_5},h_{k_6},h_{k_1},h_{k_3}) \), where \( h_j \) and \( h_k \) are siblings, \( h_{k_1} \) and \( h_{k_2} \) are the children of \( h_k \), nodes \( h_{k_4} \) and \( h_{k_5} \) are the children of \( h_{k_2} \), node \( h_{k_3} \) is a child of \( h_{k_1} \), and \( h_{k_6} \) is a child of \( h_{k_4} \);
\label{itm:penta12}
\item \( (h_j,h_{k_5},h_{k_7},h_{k_6},h_{k_1}) \), where \( h_j \) and \( h_k \) are siblings, \( h_{k_1} \) and \( h_{k_2} \) are the children of \( h_k \), nodes \( h_{k_4} \) and \( h_{k_5} \) are the children of \( h_{k_2} \), node \( h_{k_6} \) is a child of \( h_{k_4} \), and \( h_{k_7} \) is a child of \( h_{k_5} \);
\label{itm:penta13}
\item \( (h_j,h_{k_5},h_{k_7},h_{k_6},h_{k_1},h_{k_3}) \), where \( h_j \) and \( h_k \) are siblings, \( h_{k_1} \) and \( h_{k_2} \) are the children of \( h_k \), nodes \( h_{k_4} \) and \( h_{k_5} \) are the children of \( h_{k_2} \), node \( h_{k_3} \) is a child of \( h_{k_1} \), node \( h_{k_6} \) is a child of \( h_{k_4} \), and \( h_{k_7} \) is a child of \( h_{k_5} \);
\label{itm:hexa13}
\item \( (h_j,h_{k_5},h_{k_7},h_{k_8},h_{k_6},h_{k_1}) \), where \( h_j \) and \( h_k \) are siblings, \( h_{k_1} \) and \( h_{k_2} \) are the children of \( h_k \), nodes \( h_{k_4} \) and \( h_{k_5} \) are the children of \( h_{k_2} \), nodes \( h_{k_6} \) and \( h_{k_8} \) are the children of \( h_{k_4} \), and \( h_{k_7} \) is a child of \( h_{k_5} \);
\label{itm:hexa14}
\item \( (h_j,h_{k_5},h_{k_7},h_{k_8},h_{k_9},h_{k_6},h_{k_1}) \), where \( h_j \) and \( h_k \) are siblings, \( h_{k_1} \) and \( h_{k_2} \) are the children of \( h_k \), nodes \( h_{k_4} \) and \( h_{k_5} \) are the children of \( h_{k_2} \), nodes \( h_{k_6} \) and \( h_{k_8} \) are the children of \( h_{k_4} \), node \( h_{k_7} \) is a child of \( h_{k_5} \), and \( h_{k_9} \) is a child of \( h_{k_8} \);
\label{itm:hepta8}
\item \( (h_j,h_{k_5},h_{k_7},h_{k_8},h_{k_9},h_{k_{10}},h_{k_6},h_{k_1}) \), where \( h_j \) and \( h_k \) are siblings, \( h_{k_1} \) and \( h_{k_2} \) are the children of \( h_k \), nodes \( h_{k_4} \) and \( h_{k_5} \) are the children of \( h_{k_2} \), nodes \( h_{k_6} \) and \( h_{k_8} \) are the children of \( h_{k_4} \), node \( h_{k_7} \) is a child of \( h_{k_5} \), node \( h_{k_9} \) is a child of \( h_{k_8} \), and \( h_{k_{10}} \) is a child of \( h_{k_9} \);
\label{itm:octa3}
\item \( (h_j,h_{k_5},h_{k_7},h_{k_8},h_{k_6},h_{k_1},h_{k_3}) \), where \( h_j \) and \( h_k \) are siblings, \( h_{k_1} \) and \( h_{k_2} \) are the children of \( h_k \), nodes \( h_{k_4} \) and \( h_{k_5} \) are the children of \( h_{k_2} \), node \( h_{k_3} \) is a child of \( h_{k_1} \), nodes \( h_{k_6} \) and \( h_{k_8} \) are the children of \( h_{k_4} \), and \( h_{k_7} \) is a child of \( h_{k_5} \);
\label{itm:hepta9}
\item \( (h_j,h_{k_5},h_{k_7},h_{k_8},h_{k_9},h_{k_6},h_{k_1},h_{k_3}) \), where \( h_j \) and \( h_k \) are siblings, \( h_{k_1} \) and \( h_{k_2} \) are the children of \( h_k \), nodes \( h_{k_4} \) and \( h_{k_5} \) are the children of \( h_{k_2} \), node \( h_{k_3} \) is a child of \( h_{k_1} \), nodes \( h_{k_6} \) and \( h_{k_8} \) are the children of \( h_{k_4} \), node \( h_{k_7} \) is a child of \( h_{k_5} \), and \( h_{k_9} \) is a child of \( h_{k_8} \);
\label{itm:octa4}
\item \( (h_j,h_{k_5},h_{k_7},h_{k_8},h_{k_9},h_{k_{10}},h_{k_6},h_{k_1},h_{k_3}) \), where \( h_j \) and \( h_k \) are siblings, \( h_{k_1} \) and \( h_{k_2} \) are the children of \( h_k \), nodes \( h_{k_4} \) and \( h_{k_5} \) are the children of \( h_{k_2} \), node \( h_{k_3} \) is a child of \( h_{k_1} \), nodes \( h_{k_6} \) and \( h_{k_8} \) are the children of \( h_{k_4} \), node \( h_{k_7} \) is a child of \( h_{k_5} \), node \( h_{k_9} \) is a child of \( h_{k_8} \), and \( h_{k_{10}} \) is a child of \( h_{k_9} \);
\label{itm:nona1}
\item \( (h_j,h_{k_5},h_{k_8},h_{k_9},h_{k_6},h_{k_1}) \), where \( h_j \) and \( h_k \) are siblings, \( h_{k_1} \) and \( h_{k_2} \) are the children of \( h_k \), nodes \( h_{k_4} \) and \( h_{k_5} \) are the children of \( h_{k_2} \), nodes \( h_{k_6} \) and \( h_{k_8} \) are the children of \( h_{k_4} \), and \( h_{k_9} \) is a child of \( h_{k_8} \);
\label{itm:hexa15}
\item \( (h_j,h_{k_5},h_{k_8},h_{k_9},h_{k_{10}},h_{k_6},h_{k_1}) \), where \( h_j \) and \( h_k \) are siblings, \( h_{k_1} \) and \( h_{k_2} \) are the children of \( h_k \), nodes \( h_{k_4} \) and \( h_{k_5} \) are the children of \( h_{k_2} \), nodes \( h_{k_6} \) and \( h_{k_8} \) are the children of \( h_{k_4} \), node \( h_{k_9} \) is a child of \( h_{k_8} \), and \( h_{k_{10}} \) is a child of \( h_{k_9} \);
\label{itm:hepta10}
\item \( (h_j,h_{k_5},h_{k_8},h_{k_6},h_{k_1},h_{k_3}) \), where \( h_j \) and \( h_k \) are siblings, \( h_{k_1} \) and \( h_{k_2} \) are the children of \( h_k \), nodes \( h_{k_4} \) and \( h_{k_5} \) are the children of \( h_{k_2} \), node \( h_{k_3} \) is a child of \( h_{k_1} \), and \( h_{k_6} \) and \( h_{k_8} \) are the children of \( h_{k_4} \);
\label{itm:hexa16}
\item \( (h_j,h_{k_5},h_{k_8},h_{k_9},h_{k_6},h_{k_1},h_{k_3}) \), where \( h_j \) and \( h_k \) are siblings, \( h_{k_1} \) and \( h_{k_2} \) are the children of \( h_k \), nodes \( h_{k_4} \) and \( h_{k_5} \) are the children of \( h_{k_2} \), node \( h_{k_3} \) is a child of \( h_{k_1} \), nodes \( h_{k_6} \) and \( h_{k_8} \) are the children of \( h_{k_4} \), and \( h_{k_9} \) is a child of \( h_{k_8} \);
\label{itm:hepta11}
\item \( (h_j,h_{k_5},h_{k_8},h_{k_9},h_{k_{10}},h_{k_6},h_{k_1},h_{k_3}) \), where \( h_j \) and \( h_k \) are siblings, \( h_{k_1} \) and \( h_{k_2} \) are the children of \( h_k \), nodes \( h_{k_4} \) and \( h_{k_5} \) are the children of \( h_{k_2} \), node \( h_{k_3} \) is a child of \( h_{k_1} \), nodes \( h_{k_6} \) and \( h_{k_8} \) are the children of \( h_{k_4} \), node \( h_{k_9} \) is a child of \( h_{k_8} \), and \( h_{k_{10}} \) is a child of \( h_{k_9} \);
\label{itm:octa5}
\item \( (h_j,h_{k_5},h_{k_9},h_{k_6},h_{k_1},h_{k_3}) \), where \( h_j \) and \( h_k \) are siblings, \( h_{k_1} \) and \( h_{k_2} \) are the children of \( h_k \), nodes \( h_{k_4} \) and \( h_{k_5} \) are the children of \( h_{k_2} \), node \( h_{k_3} \) is a child of \( h_{k_1} \), nodes \( h_{k_6} \) and \( h_{k_8} \) are the children of \( h_{k_4} \), and \( h_{k_9} \) is a child of \( h_{k_8} \);
\label{itm:hexa17}
\item \( (h_j,h_{k_5},h_{k_9},h_{k_{10}},h_{k_6},h_{k_1}) \), where \( h_j \) and \( h_k \) are siblings, \( h_{k_1} \) and \( h_{k_2} \) are the children of \( h_k \), nodes \( h_{k_4} \) and \( h_{k_5} \) are the children of \( h_{k_2} \), nodes \( h_{k_6} \) and \( h_{k_8} \) are the children of \( h_{k_4} \), node \( h_{k_9} \) is a child of \( h_{k_8} \), and \( h_{k_{10}} \) is a child of \( h_{k_9} \);
\label{itm:hexa18}
\item \( (h_j,h_{k_5},h_{k_9},h_{k_{10}},h_{k_6},h_{k_1},h_{k_3}) \), where \( h_j \) and \( h_k \) are siblings, \( h_{k_1} \) and \( h_{k_2} \) are the children of \( h_k \), nodes \( h_{k_4} \) and \( h_{k_5} \) are the children of \( h_{k_2} \), node \( h_{k_3} \) is a child of \( h_{k_1} \), nodes \( h_{k_6} \) and \( h_{k_8} \) are the children of \( h_{k_4} \), node \( h_{k_9} \) is a child of \( h_{k_8} \), and \( h_{k_{10}} \) is a child of \( h_{k_9} \);
\label{itm:hepta12}
\item \( (h_j,h_{k_5},h_{k_9},h_{k_{10}},h_{k_{11}},h_{k_6},h_{k_1}) \), where \( h_j \) and \( h_k \) are siblings, \( h_{k_1} \) and \( h_{k_2} \) are the children of \( h_k \), nodes \( h_{k_4} \) and \( h_{k_5} \) are the children of \( h_{k_2} \), nodes \( h_{k_6} \) and \( h_{k_8} \) are the children of \( h_{k_4} \), nodes \( h_{k_9} \) and \( h_{k_{11}} \) are the children of \( h_{k_8} \), and \( h_{k_{10}} \) is a child of \( h_{k_9} \);
\label{itm:hepta13}
\item \( (h_j,h_{k_5},h_{k_9},h_{k_{10}},h_{k_{12}},h_{k_{11}},h_{k_6},h_{k_1}) \), where \( h_j \) and \( h_k \) are siblings, \( h_{k_1} \) and \( h_{k_2} \) are the children of \( h_k \), nodes \( h_{k_4} \) and \( h_{k_5} \) are the children of \( h_{k_2} \), nodes \( h_{k_6} \) and \( h_{k_8} \) are the children of \( h_{k_4} \), nodes \( h_{k_9} \) and \( h_{k_{11}} \) are the children of \( h_{k_8} \), node \( h_{k_{10}} \) is a child of \( h_{k_9} \), and \( h_{k_{12}} \) is a child of \( h_{k_{10}} \);
\label{itm:octa6}
\item \( (h_j,h_{k_5},h_{k_9},h_{k_{11}},h_{k_{13}},h_{k_6},h_{k_1}) \), where \( h_j \) and \( h_k \) are siblings, \( h_{k_1} \) and \( h_{k_2} \) are the children of \( h_k \), nodes \( h_{k_4} \) and \( h_{k_5} \) are the children of \( h_{k_2} \), nodes \( h_{k_6} \) and \( h_{k_8} \) are the children of \( h_{k_4} \), nodes \( h_{k_9} \) and \( h_{k_{11}} \) are the children of \( h_{k_8} \), and \( h_{k_{13}} \) is a child of \( h_{k_{11}} \);
\label{itm:hepta14}
\item \( (h_j,h_{k_5},h_{k_9},h_{k_{10}},h_{k_{11}},h_{k_{13}},h_{k_6},h_{k_1}) \), where \( h_j \) and \( h_k \) are siblings, \( h_{k_1} \) and \( h_{k_2} \) are the children of \( h_k \), nodes \( h_{k_4} \) and \( h_{k_5} \) are the children of \( h_{k_2} \), nodes \( h_{k_6} \) and \( h_{k_8} \) are the children of \( h_{k_4} \), nodes \( h_{k_9} \) and \( h_{k_{11}} \) are the children of \( h_{k_8} \), node \( h_{k_{10}} \) is a child of \( h_{k_9} \), and \( h_{k_{13}} \) is a child of \( h_{k_{11}} \);
\label{itm:octa7}
\item \( (h_j,h_{k_5},h_{k_9},h_{k_{10}},h_{k_{12}},h_{k_{11}},h_{k_{13}},h_{k_6},h_{k_1}) \), where \( h_j \) and \( h_k \) are siblings, \( h_{k_1} \) and \( h_{k_2} \) are the children of \( h_k \), nodes \( h_{k_4} \) and \( h_{k_5} \) are the children of \( h_{k_2} \), nodes \( h_{k_6} \) and \( h_{k_8} \) are the children of \( h_{k_4} \), nodes \( h_{k_9} \) and \( h_{k_{11}} \) are the children of \( h_{k_8} \), node \( h_{k_{10}} \) is a child of \( h_{k_9} \), node \( h_{k_{12}} \) is a child of \( h_{k_{10}} \), and \( h_{k_{13}} \) is a child of \( h_{k_{11}} \);
\label{itm:nona2}
\item \( (h_j,h_{k_5},h_{k_9},h_{k_{11}},h_{k_6},h_{k_1},h_{k_3}) \), where \( h_j \) and \( h_k \) are siblings, \( h_{k_1} \) and \( h_{k_2} \) are the children of \( h_k \), nodes \( h_{k_4} \) and \( h_{k_5} \) are the children of \( h_{k_2} \), node \( h_{k_3} \) is a child of \( h_{k_1} \), nodes \( h_{k_6} \) and \( h_{k_8} \) are the children of \( h_{k_4} \), and \( h_{k_9} \) and \( h_{k_{11}} \) are the children of \( h_{k_8} \);
\label{itm:hepta15}
\item \( (h_j,h_{k_5},h_{k_9},h_{k_{10}},h_{k_{11}},h_{k_6},h_{k_1},h_{k_3}) \), where \( h_j \) and \( h_k \) are siblings, \( h_{k_1} \) and \( h_{k_2} \) are the children of \( h_k \), nodes \( h_{k_4} \) and \( h_{k_5} \) are the children of \( h_{k_2} \), node \( h_{k_3} \) is a child of \( h_{k_1} \), nodes \( h_{k_6} \) and \( h_{k_8} \) are the children of \( h_{k_4} \), nodes \( h_{k_9} \) and \( h_{k_{11}} \) are the children of \( h_{k_8} \), and \( h_{k_{10}} \) is a child of \( h_{k_9} \);
\label{itm:octa8}
\item \( (h_j,h_{k_5},h_{k_9},h_{k_{10}},h_{k_{12}},h_{k_{11}},h_{k_6},h_{k_1},h_{k_3}) \), where \( h_j \) and \( h_k \) are siblings, \( h_{k_1} \) and \( h_{k_2} \) are the children of \( h_k \), nodes \( h_{k_4} \) and \( h_{k_5} \) are the children of \( h_{k_2} \), node \( h_{k_3} \) is a child of \( h_{k_1} \), nodes \( h_{k_6} \) and \( h_{k_8} \) are the children of \( h_{k_4} \), nodes \( h_{k_9} \) and \( h_{k_{11}} \) are the children of \( h_{k_8} \), node \( h_{k_{10}} \) is a child of \( h_{k_9} \), and \( h_{k_{12}} \) is a child of \( h_{k_{10}} \);
\label{itm:nona3}
\item \( (h_j,h_{k_5},h_{k_9},h_{k_{11}},h_{k_{13}},h_{k_6},h_{k_1},h_{k_3}) \), where \( h_j \) and \( h_k \) are siblings, \( h_{k_1} \) and \( h_{k_2} \) are the children of \( h_k \), nodes \( h_{k_4} \) and \( h_{k_5} \) are the children of \( h_{k_2} \), node \( h_{k_3} \) is a child of \( h_{k_1} \), nodes \( h_{k_6} \) and \( h_{k_8} \) are the children of \( h_{k_4} \), nodes \( h_{k_9} \) and \( h_{k_{11}} \) are the children of \( h_{k_8} \), and \( h_{k_{13}} \) is a child of \( h_{k_{11}} \);
\label{itm:octa9}
\item \( (h_j,h_{k_5},h_{k_9},h_{k_{10}},h_{k_{11}},h_{k_{13}},h_{k_6},h_{k_1},h_{k_3}) \), where \( h_j \) and \( h_k \) are siblings, \( h_{k_1} \) and \( h_{k_2} \) are the children of \( h_k \), nodes \( h_{k_4} \) and \( h_{k_5} \) are the children of \( h_{k_2} \), node \( h_{k_3} \) is a child of \( h_{k_1} \), nodes \( h_{k_6} \) and \( h_{k_8} \) are the children of \( h_{k_4} \), nodes \( h_{k_9} \) and \( h_{k_{11}} \) are the children of \( h_{k_8} \), node \( h_{k_{10}} \) is a child of \( h_{k_9} \), and \( h_{k_{13}} \) is a child of \( h_{k_{11}} \);
\label{itm:nona4}
\item \( (h_j,h_{k_5},h_{k_9},h_{k_{10}},h_{k_{12}},h_{k_{11}},h_{k_{13}},h_{k_6},h_{k_1},h_{k_3}) \), where \( h_j \) and \( h_k \) are siblings, \( h_{k_1} \) and \( h_{k_2} \) are the children of \( h_k \), nodes \( h_{k_4} \) and \( h_{k_5} \) are the children of \( h_{k_2} \), node \( h_{k_3} \) is a child of \( h_{k_1} \), nodes \( h_{k_6} \) and \( h_{k_8} \) are the children of \( h_{k_4} \), nodes \( h_{k_9} \) and \( h_{k_{11}} \) are the children of \( h_{k_8} \), node \( h_{k_{10}} \) is a child of \( h_{k_9} \), node \( h_{k_{12}} \) is a child of \( h_{k_{10}} \), and \( h_{k_{13}} \) is a child of \( h_{k_{11}} \);
\label{itm:deca1}
\item \( (h_j,h_{k_5},h_{k_9},h_{k_{10}},h_{k_{13}},h_{k_6},h_{k_1}) \), where \( h_j \) and \( h_k \) are siblings, \( h_{k_1} \) and \( h_{k_2} \) are the children of \( h_k \), nodes \( h_{k_4} \) and \( h_{k_5} \) are the children of \( h_{k_2} \), nodes \( h_{k_6} \) and \( h_{k_8} \) are the children of \( h_{k_4} \), nodes \( h_{k_9} \) and \( h_{k_{11}} \) are the children of \( h_{k_8} \), node \( h_{k_{10}} \) is a child of \( h_{k_9} \), and \( h_{k_{13}} \) is a child of \( h_{k_{11}} \);
\label{itm:hepta16}
\item \( (h_j,h_{k_5},h_{k_9},h_{k_{13}},h_{k_6},h_{k_1},h_{k_3}) \), where \( h_j \) and \( h_k \) are siblings, \( h_{k_1} \) and \( h_{k_2} \) are the children of \( h_k \), nodes \( h_{k_4} \) and \( h_{k_5} \) are the children of \( h_{k_2} \), node \( h_{k_3} \) is a child of \( h_{k_1} \), nodes \( h_{k_6} \) and \( h_{k_8} \) are the children of \( h_{k_4} \), nodes \( h_{k_9} \) and \( h_{k_{11}} \) are the children of \( h_{k_8} \), and \( h_{k_{13}} \) is a child of \( h_{k_{11}} \);
\label{itm:hepta17}
\item \( (h_j,h_{k_5},h_{k_9},h_{k_{10}},h_{k_{13}},h_{k_6},h_{k_1},h_{k_3}) \), where \( h_j \) and \( h_k \) are siblings, \( h_{k_1} \) and \( h_{k_2} \) are the children of \( h_k \), nodes \( h_{k_4} \) and \( h_{k_5} \) are the children of \( h_{k_2} \), node \( h_{k_3} \) is a child of \( h_{k_1} \), nodes \( h_{k_6} \) and \( h_{k_8} \) are the children of \( h_{k_4} \), nodes \( h_{k_9} \) and \( h_{k_{11}} \) are the children of \( h_{k_8} \), node \( h_{k_{10}} \) is a child of \( h_{k_9} \), and \( h_{k_{13}} \) is a child of \( h_{k_{11}} \);
\label{itm:octa10}
\item \( (h_j,h_{k_5},h_{k_9},h_{k_{10}},h_{k_{14}},h_{k_{13}},h_{k_6},h_{k_1}) \), where \( h_j \) and \( h_k \) are siblings, \( h_{k_1} \) and \( h_{k_2} \) are the children of \( h_k \), nodes \( h_{k_4} \) and \( h_{k_5} \) are the children of \( h_{k_2} \), nodes \( h_{k_6} \) and \( h_{k_8} \) are the children of \( h_{k_4} \), nodes \( h_{k_9} \) and \( h_{k_{11}} \) are the children of \( h_{k_8} \), node \( h_{k_{10}} \) is a child of \( h_{k_9} \), and \( h_{k_{13}} \) and \( h_{k_{14}} \) are the children of \( h_{k_{11}} \);
\label{itm:octa11}
\item \( (h_j,h_{k_5},h_{k_9},h_{k_{14}},h_{k_{15}},h_{k_{13}},h_{k_6},h_{k_1}) \), where \( h_j \) and \( h_k \) are siblings, \( h_{k_1} \) and \( h_{k_2} \) are the children of \( h_k \), nodes \( h_{k_4} \) and \( h_{k_5} \) are the children of \( h_{k_2} \), nodes \( h_{k_6} \) and \( h_{k_8} \) are the children of \( h_{k_4} \), nodes \( h_{k_9} \) and \( h_{k_{11}} \) are the children of \( h_{k_8} \), nodes \( h_{k_{13}} \) and \( h_{k_{14}} \) are the children of \( h_{k_{11}} \), and \( h_{k_{15}} \) is a child of \( h_{k_{14}} \);
\label{itm:octa12}
\item \( (h_j,h_{k_5},h_{k_9},h_{k_{10}},h_{k_{14}},h_{k_{15}},h_{k_{13}},h_{k_6},h_{k_1}) \), where \( h_j \) and \( h_k \) are siblings, \( h_{k_1} \) and \( h_{k_2} \) are the children of \( h_k \), nodes \( h_{k_4} \) and \( h_{k_5} \) are the children of \( h_{k_2} \), nodes \( h_{k_6} \) and \( h_{k_8} \) are the children of \( h_{k_4} \), nodes \( h_{k_9} \) and \( h_{k_{11}} \) are the children of \( h_{k_8} \), node \( h_{k_{10}} \) is a child of \( h_{k_9} \), nodes \( h_{k_{13}} \) and \( h_{k_{14}} \) are the children of \( h_{k_{11}} \), and \( h_{k_{15}} \) is a child of \( h_{k_{14}} \);
\label{itm:nona5}
\item \( (h_j,h_{k_5},h_{k_9},h_{k_{14}},h_{k_{15}},h_{k_{16}},h_{k_{13}},h_{k_6},h_{k_1}) \), where \( h_j \) and \( h_k \) are siblings, \( h_{k_1} \) and \( h_{k_2} \) are the children of \( h_k \), nodes \( h_{k_4} \) and \( h_{k_5} \) are the children of \( h_{k_2} \), nodes \( h_{k_6} \) and \( h_{k_8} \) are the children of \( h_{k_4} \), nodes \( h_{k_9} \) and \( h_{k_{11}} \) are the children of \( h_{k_8} \), nodes \( h_{k_{13}} \) and \( h_{k_{14}} \) are the children of \( h_{k_{11}} \), node \( h_{k_{15}} \) is a child of \( h_{k_{14}} \), and \( h_{k_{16}} \) is a child of \( h_{k_{15}} \);
\label{itm:nona6}
\item \( (h_j,h_{k_5},h_{k_9},h_{k_{10}},h_{k_{14}},h_{k_{15}},h_{k_{16}},h_{k_{13}},h_{k_6},h_{k_1}) \), where \( h_j \) and \( h_k \) are siblings, \( h_{k_1} \) and \( h_{k_2} \) are the children of \( h_k \), nodes \( h_{k_4} \) and \( h_{k_5} \) are the children of \( h_{k_2} \), nodes \( h_{k_6} \) and \( h_{k_8} \) are the children of \( h_{k_4} \), nodes \( h_{k_9} \) and \( h_{k_{11}} \) are the children of \( h_{k_8} \), node \( h_{k_{10}} \) is a child of \( h_{k_9} \), nodes \( h_{k_{13}} \) and \( h_{k_{14}} \) are the children of \( h_{k_{11}} \), node \( h_{k_{15}} \) is a child of \( h_{k_{14}} \), and \( h_{k_{16}} \) is a child of \( h_{k_{15}} \);
\label{itm:deca2}
\item \( (h_j,h_{k_5},h_{k_9},h_{k_{14}},h_{k_{13}},h_{k_6},h_{k_1},h_{k_3}) \), where \( h_j \) and \( h_k \) are siblings, \( h_{k_1} \) and \( h_{k_2} \) are the children of \( h_k \), nodes \( h_{k_4} \) and \( h_{k_5} \) are the children of \( h_{k_2} \), node \( h_{k_3} \) is a child of \( h_{k_1} \), nodes \( h_{k_6} \) and \( h_{k_8} \) are the children of \( h_{k_4} \), nodes \( h_{k_9} \) and \( h_{k_{11}} \) are the children of \( h_{k_8} \), and \( h_{k_{13}} \) and \( h_{k_{14}} \) are the children of \( h_{k_{11}} \);
\label{itm:octa13}
\item \( (h_j,h_{k_5},h_{k_9},h_{k_{10}},h_{k_{14}},h_{k_{13}},h_{k_6},h_{k_1},h_{k_3}) \), where \( h_j \) and \( h_k \) are siblings, \( h_{k_1} \) and \( h_{k_2} \) are the children of \( h_k \), nodes \( h_{k_4} \) and \( h_{k_5} \) are the children of \( h_{k_2} \), node \( h_{k_3} \) is a child of \( h_{k_1} \), nodes \( h_{k_6} \) and \( h_{k_8} \) are the children of \( h_{k_4} \), nodes \( h_{k_9} \) and \( h_{k_{11}} \) are the children of \( h_{k_8} \), node \( h_{k_{10}} \) is a child of \( h_{k_9} \), and \( h_{k_{13}} \) and \( h_{k_{14}} \) are the children of \( h_{k_{11}} \);
\label{itm:nona7}
\item \( (h_j,h_{k_5},h_{k_9},h_{k_{14}},h_{k_{15}},h_{k_{13}},h_{k_6},h_{k_1},h_{k_3}) \), where \( h_j \) and \( h_k \) are siblings, \( h_{k_1} \) and \( h_{k_2} \) are the children of \( h_k \), nodes \( h_{k_4} \) and \( h_{k_5} \) are the children of \( h_{k_2} \), node \( h_{k_3} \) is a child of \( h_{k_1} \), nodes \( h_{k_6} \) and \( h_{k_8} \) are the children of \( h_{k_4} \), nodes \( h_{k_9} \) and \( h_{k_{11}} \) are the children of \( h_{k_8} \), nodes \( h_{k_{13}} \) and \( h_{k_{14}} \) are the children of \( h_{k_{11}} \), and \( h_{k_{15}} \) is a child of \( h_{k_{14}} \);
\label{itm:nona8}
\item \( (h_j,h_{k_5},h_{k_9},h_{k_{10}},h_{k_{14}},h_{k_{15}},h_{k_{13}},h_{k_6},h_{k_1},h_{k_3}) \), where \( h_j \) and \( h_k \) are siblings, \( h_{k_1} \) and \( h_{k_2} \) are the children of \( h_k \), nodes \( h_{k_4} \) and \( h_{k_5} \) are the children of \( h_{k_2} \), node \( h_{k_3} \) is a child of \( h_{k_1} \), nodes \( h_{k_6} \) and \( h_{k_8} \) are the children of \( h_{k_4} \), nodes \( h_{k_9} \) and \( h_{k_{11}} \) are the children of \( h_{k_8} \), node \( h_{k_{10}} \) is a child of \( h_{k_9} \), nodes \( h_{k_{13}} \) and \( h_{k_{14}} \) are the children of \( h_{k_{11}} \), and \( h_{k_{15}} \) is a child of \( h_{k_{14}} \);
\label{itm:deca3}
\item \( (h_j,h_{k_5},h_{k_9},h_{k_{14}},h_{k_{15}},h_{k_{16}},h_{k_{13}},h_{k_6},h_{k_1},h_{k_3}) \), where \( h_j \) and \( h_k \) are siblings, \( h_{k_1} \) and \( h_{k_2} \) are the children of \( h_k \), nodes \( h_{k_4} \) and \( h_{k_5} \) are the children of \( h_{k_2} \), node \( h_{k_3} \) is a child of \( h_{k_1} \), nodes \( h_{k_6} \) and \( h_{k_8} \) are the children of \( h_{k_4} \), nodes \( h_{k_9} \) and \( h_{k_{11}} \) are the children of \( h_{k_8} \), nodes \( h_{k_{13}} \) and \( h_{k_{14}} \) are the children of \( h_{k_{11}} \), node \( h_{k_{15}} \) is a child of \( h_{k_{14}} \), and \( h_{k_{16}} \) is a child of \( h_{k_{15}} \);
\label{itm:deca4}
\item \( (h_j,h_{k_5},h_{k_9},h_{k_{10}},h_{k_{14}},h_{k_{15}},h_{k_{16}},h_{k_{13}},h_{k_6},h_{k_1},h_{k_3}) \), where \( h_j \) and \( h_k \) are siblings, \( h_{k_1} \) and \( h_{k_2} \) are the children of \( h_k \), nodes \( h_{k_4} \) and \( h_{k_5} \) are the children of \( h_{k_2} \), node \( h_{k_3} \) is a child of \( h_{k_1} \), nodes \( h_{k_6} \) and \( h_{k_8} \) are the children of \( h_{k_4} \), nodes \( h_{k_9} \) and \( h_{k_{11}} \) are the children of \( h_{k_8} \), node \( h_{k_{10}} \) is a child of \( h_{k_9} \), nodes \( h_{k_{13}} \) and \( h_{k_{14}} \) are the children of \( h_{k_{11}} \), node \( h_{k_{15}} \) is a child of \( h_{k_{14}} \), and \( h_{k_{16}} \) is a child of \( h_{k_{15}} \);
\label{itm:undeca1}
\item \( (h_j,h_{k_4},h_{k_6},h_{k_3}) \), where \( h_j \) and \( h_k \) are siblings, \( h_{k_1} \) is a child of \( h_k \), nodes \( h_{k_3} \) and \( h_{k_4} \) are the children of \( h_{k_1} \), and \( h_{k_6} \) is a child of \( h_{k_4} \);
\label{itm:penta14}
\item \( (h_j,h_{k_4},h_{k_6},h_{k_8},h_{k_3}) \), where \( h_j \) and \( h_k \) are siblings, \( h_{k_1} \) is a child of \( h_k \), nodes \( h_{k_3} \) and \( h_{k_4} \) are the children of \( h_{k_1} \), node \( h_{k_6} \) is a child of \( h_{k_4} \), and \( h_{k_8} \) is a child of \( h_{k_6} \);
\label{itm:hexa19}
\item \( (h_j,h_{j_1},h_{k_2},h_{k_3}) \), where \( h_j \) and \( h_k \) are siblings, \( h_{j_1} \) is a child of \( h_j \), nodes \( h_{k_1} \) and \( h_{k_2} \) are the children of \( h_k \), and \( h_{k_3} \) is a child of \( h_{k_1} \);
\label{itm:tetra10}
\item \( (h_j,h_{k_2},h_{k_4},h_{k_3}) \), where \( h_j \) and \( h_k \) are siblings, \( h_{k_1} \) and \( h_{k_2} \) are the children of \( h_k \), node \( h_{k_3} \) is a child of \( h_{k_1} \), and \( h_{k_4} \) is a child of \( h_{k_2} \);
\label{itm:tetra11}
\item \( (h_j,h_{j_1},h_{k_2},h_{k_4},h_{k_3}) \), where \( h_j \) and \( h_k \) are siblings, \( h_{j_1} \) is a child of \( h_j \), nodes \( h_{k_1} \) and \( h_{k_2} \) are the children of \( h_k \), node \( h_{k_3} \) is a child of \( h_{k_1} \), and \( h_{k_4} \) is a child of \( h_{k_2} \);
\label{itm:penta15}
\item \( (h_j,h_{k_1},h_{k_4},h_{k_6},h_{k_5}) \), where \( h_j \) and \( h_k \) are siblings, \( h_{k_1} \) and \( h_{k_2} \) are the children of \( h_k \), nodes \( h_{k_4} \) and \( h_{k_5} \) are the children of \( h_{k_2} \), and \( h_{k_6} \) is a child of \( h_{k_4} \);
\label{itm:penta16}
\item \( (h_j,h_{k_1},h_{k_4},h_{k_6},h_{k_8},h_{k_5}) \), where \( h_j \) and \( h_k \) are siblings, \( h_{k_1} \) and \( h_{k_2} \) are the children of \( h_k \), nodes \( h_{k_4} \) and \( h_{k_5} \) are the children of \( h_{k_2} \), node \( h_{k_6} \) is a child of \( h_{k_4} \), and \( h_{k_8} \) is a child of \( h_{k_6} \);
\label{itm:hexa20}
\item \( (h_j,h_{k_1},h_{k_3},h_{k_4},h_{k_5}) \), where \( h_j \) and \( h_k \) are siblings, \( h_{k_1} \) and \( h_{k_2} \) are the children of \( h_k \), nodes \( h_{k_4} \) and \( h_{k_5} \) are the children of \( h_{k_2} \), and \( h_{k_3} \) is a child of \( h_{k_1} \);
\label{itm:penta17}
\item \( (h_j,h_{k_1},h_{k_3},h_{k_4},h_{k_6},h_{k_5}) \), where \( h_j \) and \( h_k \) are siblings, \( h_{k_1} \) and \( h_{k_2} \) are the children of \( h_k \), nodes \( h_{k_4} \) and \( h_{k_5} \) are the children of \( h_{k_2} \), node \( h_{k_3} \) is a child of \( h_{k_1} \), and \( h_{k_6} \) is a child of \( h_{k_4} \);
\label{itm:hexa21}
\item \( (h_j,h_{k_1},h_{k_3},h_{k_4},h_{k_6},h_{k_8},h_{k_5}) \), where \( h_j \) and \( h_k \) are siblings, \( h_{k_1} \) and \( h_{k_2} \) are the children of \( h_k \), nodes \( h_{k_4} \) and \( h_{k_5} \) are the children of \( h_{k_2} \), node \( h_{k_3} \) is a child of \( h_{k_1} \), node \( h_{k_6} \) is a child of \( h_{k_4} \), and \( h_{k_8} \) is a child of \( h_{k_6} \);
\label{itm:hepta18}
\item \( (h_j,h_{j_1},h_{k_1},h_{k_4},h_{k_5}) \), where \( h_j \) and \( h_k \) are siblings, node \( h_{j_1} \) is a child of \( h_j \), nodes \( h_{k_1} \) and \( h_{k_2} \) are the children of \( h_k \), and \( h_{k_4} \) and \( h_{k_5} \) are the children of \( h_{k_2} \);
\label{itm:penta18}
\item \( (h_j,h_{j_1},h_{k_1},h_{k_4},h_{k_6},h_{k_5}) \), where \( h_j \) and \( h_k \) are siblings, \( h_{j_1} \) is a child of \( h_j \), nodes \( h_{k_1} \) and \( h_{k_2} \) are the children of \( h_k \), nodes \( h_{k_4} \) and \( h_{k_5} \) are the children of \( h_{k_2} \), and \( h_{k_6} \) is a child of \( h_{k_4} \);
\label{itm:hexa22}
\item \( (h_j,h_{j_1},h_{k_1},h_{k_4},h_{k_6},h_{k_8},h_{k_5}) \), where \( h_j \) and \( h_k \) are siblings, \( h_{j_1} \) is a child of \( h_j \), nodes \( h_{k_1} \) and \( h_{k_2} \) are the children of \( h_k \), nodes \( h_{k_4} \) and \( h_{k_5} \) are the children of \( h_{k_2} \), node \( h_{k_6} \) is a child of \( h_{k_4} \), and \( h_{k_8} \) is a child of \( h_{k_6} \);
\label{itm:hepta19}
\item \( (h_j,h_{j_1},h_{k_1},h_{k_3},h_{k_4},h_{k_5}) \), where \( h_j \) and \( h_k \) are siblings, \( h_{j_1} \) is a child of \( h_j \), nodes \( h_{k_1} \) and \( h_{k_2} \) are the children of \( h_k \), nodes \( h_{k_4} \) and \( h_{k_5} \) are the children of \( h_{k_2} \), and \( h_{k_3} \) is a child of \( h_{k_1} \);
\label{itm:hexa23}
\item \( (h_j,h_{j_1},h_{k_1},h_{k_3},h_{k_4},h_{k_6},h_{k_5}) \), where \( h_j \) and \( h_k \) are siblings, \( h_{j_1} \) is a child of \( h_j \), nodes \( h_{k_1} \) and \( h_{k_2} \) are the children of \( h_k \), nodes \( h_{k_4} \) and \( h_{k_5} \) are the children of \( h_{k_2} \), node \( h_{k_3} \) is a child of \( h_{k_1} \), and \( h_{k_6} \) is a child of \( h_{k_4} \);
\label{itm:hepta20}
\item \( (h_j,h_{j_1},h_{k_1},h_{k_3},h_{k_4},h_{k_6},h_{k_8},h_{k_5}) \), where \( h_j \) and \( h_k \) are siblings, \( h_{j_1} \) is a child of \( h_j \), nodes \( h_{k_1} \) and \( h_{k_2} \) are the children of \( h_k \), nodes \( h_{k_4} \) and \( h_{k_5} \) are the children of \( h_{k_2} \), node \( h_{k_3} \) is a child of \( h_{k_1} \), node \( h_{k_6} \) is a child of \( h_{k_4} \), and \( h_{k_8} \) is a child of \( h_{k_6} \);
\label{itm:octa14}
\item \( (h_j,h_{k_4},h_{k_6},h_{k_5},h_{k_3}) \), where \( h_j \) and \( h_k \) are siblings, \( h_{k_1} \) and \( h_{k_2} \) are the children of \( h_k \), node \( h_{k_3} \) is a child of \( h_{k_1} \), nodes \( h_{k_4} \) and \( h_{k_5} \) are the children of \( h_{k_2} \), and \( h_{k_6} \) is a child of \( h_{k_4} \);
\label{itm:penta19}
\item \( (h_j,h_{k_4},h_{k_6},h_{k_8},h_{k_5},h_{k_3}) \), where \( h_j \) and \( h_k \) are siblings, \( h_{k_1} \) and \( h_{k_2} \) are the children of \( h_k \), node \( h_{k_3} \) is a child of \( h_{k_1} \), nodes \( h_{k_4} \) and \( h_{k_5} \) are the children of \( h_{k_2} \), node \( h_{k_6} \) is a child of \( h_{k_4} \), and \( h_{k_8} \) is a child of \( h_{k_6} \);
\label{itm:hexa24}
\item \( (h_j,h_{k_5},h_{k_9},h_{k_6},h_{k_4},h_{k_3}) \), where \( h_j \) and \( h_k \) are siblings, \( h_{k_1} \) and \( h_{k_2} \) are the children of \( h_k \), nodes \( h_{k_3} \) and \( h_{k_4} \) are the children of \( h_{k_1} \), nodes \( h_{k_5} \) and \( h_{k_6} \) are the children of \( h_{k_2} \), and \( h_{k_9} \) is a child of \( h_{k_5} \);
\label{itm:hexa25}
\item \( (h_j,h_{k_5},h_{k_9},h_{k_{10}},h_{k_6},h_{k_4},h_{k_3}) \), where \( h_j \) and \( h_k \) are siblings, \( h_{k_1} \) and \( h_{k_2} \) are the children of \( h_k \), nodes \( h_{k_3} \) and \( h_{k_4} \) are the children of \( h_{k_1} \), nodes \( h_{k_5} \) and \( h_{k_6} \) are the children of \( h_{k_2} \), node \( h_{k_9} \) is a child of \( h_{k_5} \), and \( h_{k_{10}} \) is a child of \( h_{k_9} \);
\label{itm:hepta21}
\item \( (h_j,h_{k_5},h_{k_6},h_{k_4},h_{k_7},h_{k_3}) \), where \( h_j \) and \( h_k \) are siblings, \( h_{k_1} \) and \( h_{k_2} \) are the children of \( h_k \), nodes \( h_{k_3} \) and \( h_{k_4} \) are the children of \( h_{k_1} \), nodes \( h_{k_5} \) and \( h_{k_6} \) are the children of \( h_{k_2} \), and \( h_{k_7} \) is a child of \( h_{k_4} \);
\label{itm:hexa26}
\item \( (h_j,h_{k_5},h_{k_9},h_{k_6},h_{k_4},h_{k_7},h_{k_3}) \), where \( h_j \) and \( h_k \) are siblings, \( h_{k_1} \) and \( h_{k_2} \) are the children of \( h_k \), nodes \( h_{k_3} \) and \( h_{k_4} \) are the children of \( h_{k_1} \), nodes \( h_{k_5} \) and \( h_{k_6} \) are the children of \( h_{k_2} \), node \( h_{k_7} \) is a child of \( h_{k_4} \), and \( h_{k_9} \) is a child of \( h_{k_5} \);
\label{itm:hepta22}
\item \( (h_j,h_{k_5},h_{k_9},h_{k_{10}},h_{k_6},h_{k_4},h_{k_7},h_{k_3}) \), where \( h_j \) and \( h_k \) are siblings, \( h_{k_1} \) and \( h_{k_2} \) are the children of \( h_k \), nodes \( h_{k_3} \) and \( h_{k_4} \) are the children of \( h_{k_1} \), nodes \( h_{k_5} \) and \( h_{k_6} \) are the children of \( h_{k_2} \), node \( h_{k_7} \) is a child of \( h_{k_4} \), node \( h_{k_9} \) is a child of \( h_{k_5} \), and \( h_{k_{10}} \) is a child of \( h_{k_9} \);
\label{itm:octa15}
\item \( (h_j,h_{k_5},h_{k_6},h_{k_4},h_{k_7},h_{k_8},h_{k_3}) \), where \( h_j \) and \( h_k \) are siblings, \( h_{k_1} \) and \( h_{k_2} \) are the children of \( h_k \), nodes \( h_{k_3} \) and \( h_{k_4} \) are the children of \( h_{k_1} \), nodes \( h_{k_5} \) and \( h_{k_6} \) are the children of \( h_{k_2} \), node \( h_{k_7} \) is a child of \( h_{k_4} \), and \( h_{k_8} \) is a child of \( h_{k_7} \);
\label{itm:hepta23}
\item \( (h_j,h_{k_5},h_{k_9},h_{k_6},h_{k_4},h_{k_7},h_{k_8},h_{k_3}) \), where \( h_j \) and \( h_k \) are siblings, \( h_{k_1} \) and \( h_{k_2} \) are the children of \( h_k \), nodes \( h_{k_3} \) and \( h_{k_4} \) are the children of \( h_{k_1} \), nodes \( h_{k_5} \) and \( h_{k_6} \) are the children of \( h_{k_2} \), node \( h_{k_7} \) is a child of \( h_{k_4} \), node \( h_{k_9} \) is a child of \( h_{k_5} \), and \( h_{k_8} \) is a child of \( h_{k_7} \);
\label{itm:octa16}
\item \( (h_j,h_{k_5},h_{k_9},h_{k_{10}},h_{k_6},h_{k_4},h_{k_7},h_{k_8},h_{k_3}) \), where \( h_j \) and \( h_k \) are siblings, \( h_{k_1} \) and \( h_{k_2} \) are the children of \( h_k \), nodes \( h_{k_3} \) and \( h_{k_4} \) are the children of \( h_{k_1} \), nodes \( h_{k_5} \) and \( h_{k_6} \) are the children of \( h_{k_2} \), node \( h_{k_7} \) is a child of \( h_{k_4} \), node \( h_{k_9} \) is a child of \( h_{k_5} \), node \( h_{k_8} \) is a child of \( h_{k_7} \), and \( h_{k_{10}} \) is a child of \( h_{k_9} \);
\label{itm:nona9}
\end{enumerate}

\end{document}